\documentclass[11pt]{article}
\usepackage[margin=1in]{geometry}
\usepackage{fullpage}
\usepackage{times}
\usepackage[T1]{fontenc}
\usepackage{graphicx}
\usepackage{amsfonts,amsmath,amsthm,amssymb,dsfont,mathtools}
\usepackage{xfrac,nicefrac}
\usepackage{mathdots}
\usepackage{bm}
\usepackage{url}
\usepackage{paralist}
\usepackage[shortlabels]{enumitem}
\usepackage[normalem]{ulem}
\usepackage{xspace}
\xspaceaddexceptions{]\}}
\usepackage{comment}
\usepackage{tabu}
\usepackage{framed}
\usepackage{float,wrapfig}
\usepackage{subfigure}
\usepackage{tikz}
\usepackage{algpseudocode} 
\usepackage[usenames,dvipsnames]{pstricks}
\usepackage[linesnumbered,boxed,ruled,vlined]{algorithm2e}
\usepackage{thm-restate}

\usepackage[colorlinks,citecolor=blue,linkcolor=blue,urlcolor=red,pagebackref]{hyperref}

\usepackage[capitalise]{cleveref}

\theoremstyle{plain}
\newtheorem{prob}{Problem}
\newtheorem{open}{Open Question}
\newtheorem{theorem}{Theorem}[section]
\newtheorem{lemma}[theorem]{Lemma}

\newtheorem{cor}[theorem]{Corollary}

\theoremstyle{definition}
\newtheorem{definition}[theorem]{Definition}
\newtheorem{remark}[theorem]{Remark}

 \def\ShowAuthNotes{1}
 \ifnum\ShowAuthNotes=1
 \newcommand{\authnote}[2]{\ \\ \textcolor{red}{\parbox{0.9\linewidth}{[{\footnotesize {\bf #1:} { {#2}}}]}}\newline}
 \else
 \newcommand{\authnote}[2]{}
 \fi

\renewcommand{\epsilon}{\varepsilon}
\newcommand{\eps}{\varepsilon}
\renewcommand{\Pr}{\operatorname*{\mathbf{Pr}}}
\newcommand{\Ex}{\operatorname*{\mathbb{E}}}
\newcommand{\bin}{\operatorname*{\mathrm{bin}}}
\newcommand{\low}{\operatorname*{\mathrm{low}}}
\newcommand{\Var}{\operatorname*{\mathrm{Var}}}

\newcommand{\poly}{\operatorname{\mathrm{poly}}}
\newcommand{\polylog}{\poly\log}
\newcommand{\F}{\mathbb{F}}

\newcommand{\N}{\mathbb{N}}

\newcommand{\Z}{\mathbb{Z}}

\newcommand{\dd}{\mathinner{.\,.}}
\def\tO{\widetilde{O}}
\newcommand{\D}{\Delta}
\newcommand{\OO}{\widetilde{O}}
\newcommand{\COUNT}{\mbox{\sc count}}
\newcommand{\Meq}{M_{\mbox{\scriptsize\rm eq}}}

\title{Faster Algorithms for Text-to-Pattern Hamming Distances\thanks{Timothy Chan is partially supported by NSF grant CCF-2224271.   Ce Jin is partially supported by NSF grants CCF-2129139 and CCF-2127597. Virginia Vassilevska Williams and Yinzhan Xu are partially supported by NSF grants CCF-2129139 and CCF-2330048 and BSF Grant 2020356.}}
\author{
Timothy M. Chan\thanks{tmc@illinois.edu}\\UIUC \and
Ce Jin\thanks{cejin@mit.edu}\\MIT
\and 
Virginia Vassilevska Williams\thanks{virgi@mit.edu}\\MIT
\and 
Yinzhan Xu\thanks{xyzhan@mit.edu}\\MIT
}
\date{}

\begin{document}

	\setcounter{page}{0} \clearpage
	\maketitle
	\thispagestyle{empty}
\begin{abstract}
We study the classic \emph{Text-to-Pattern Hamming Distances} problem:  given a pattern $P$ of length $m$ and a text $T$ of length $n$, both over a polynomial-size alphabet, compute the Hamming distance between $P$ and $T[i \mathinner{.\,.} i + m-1]$ for every shift $i$, under the standard Word-RAM model with $\Theta(\log n)$-bit words. 

\begin{itemize}
    \item 

    We provide an $O(n\sqrt{m})$ time Las Vegas randomized algorithm for this problem, beating the decades-old $O(n \sqrt{m \log m})$ running time [Abrahamson, SICOMP 1987]. We also obtain a deterministic algorithm, with a slightly higher $O(n \sqrt{m} (\log m \log \log m)^{1/4})$ running time. 
    
    Our randomized algorithm extends to the $k$-bounded setting, with running time  $O\big (n + \frac{nk}{\sqrt{m}}\big )$, removing all the extra logarithmic factors from earlier algorithms [Gawrychowski and Uzna\'{n}ski, ICALP 2018;  Chan, Golan, Kociumaka, Kopelowitz and Porat, STOC 2020].
    
    \item For the $(1+\eps)$-approximate  version of Text-to-Pattern Hamming Distances, we give an $\widetilde{O}(\epsilon^{-0.93}n)$ time Monte Carlo randomized algorithm (where $\widetilde{O}$ hides poly-logarithmic factors), beating the previous $\widetilde{O}(\epsilon^{-1} n)$ running time [Kopelowitz and Porat, FOCS 2015; Kopelowitz and Porat, SOSA 2018].
\end{itemize}

Our approximation algorithm exploits a connection with $3$SUM, and uses a combination of Fredman's trick, equality matrix product, and random sampling; in particular, we obtain new results on approximate counting versions of $3$SUM and Exact Triangle, which may be of independent interest. Our exact algorithms use a novel combination of hashing, bit-packed FFT, and recursion; in particular, we obtain a faster algorithm for computing the sumset of two integer sets, in the regime when the universe size is close to quadratic in the number of elements.

We also prove a fine-grained equivalence between the exact Text-to-Pattern Hamming Distances problem and a range-restricted, counting version of $3$SUM.

	\end{abstract}
	\newpage

\maketitle

\section{Introduction}

In this paper,
we study one of the most basic problems about string matching, the classic \emph{Text-to-Pattern Hamming Distances} problem (also known as Sliding Window Hamming Distances, or String Matching with Mismatches):  given a pattern $P$ of length $m$ and a text $T$ of length $n$ over an alphabet of size $\sigma$, compute the Hamming distance (i.e., the number of mismatches) between $P$ and $T[i \mathinner{.\,.} i + m-1]$ for every shift $i$.

Fischer and Paterson's seminal work \cite{fispat} gave an algorithm running in $O(\sigma n\log m)$ time\footnote{We consider the Word-RAM model with $\Theta(\log n)$-bit words throughout the paper. } by reducing it to convolution or polynomial multiplication, which can be solved using the Fast Fourier Transform (FFT); this is the fastest known algorithm for small $\sigma$. For arbitrary alphabet size,  well-known work by Abrahamson \cite{Abrahamson87} described an $O(n\sqrt{m}\mathop{\rm polylog} n)$ time algorithm
for a family of generalized string matching problems;
for Text-to-Pattern Hamming Distances, the time bound is $O(n\sqrt{m\log m})$.  Abrahamson's algorithm
was perhaps the first example of a string algorithm with ``intermediate complexity'' between linear and quadratic (ignoring logs).
  A fine-grained reduction attributed to Indyk (see \cite{indyk-lb}) shows that no combinatorial algorithm for Text-to-Pattern Hamming Distances can run in $O(nm^{1/2 - \delta})$ time for an arbitrarily small constant $\delta>0$, under the combinatorial Boolean Matrix Multiplication Hypothesis.\footnote{While the notion of ``combinatorial'' is not well-defined, the typical notion of a combinatorial algorithm is one that does not use fast matrix multiplication. The combinatorial Boolean Matrix Multiplication hypothesis states that there is no combinatorial algorithm that multiplies two $n\times n$ Boolean matrices in $O(n^{3-\eps})$ time, for any constant $\eps>0$, in the word-RAM model with $\Theta(\log n)$ bit words.} This suggests that Abrahamson's algorithm might be optimal up to sub-polynomial factors, at least for combinatorial algorithms. 

However, so far not even poly-logarithmic improvements of Abrahamson's algorithm have been reported. This is not due to a lack of interest. In fact, many algorithms are designed to shave logarithmic factors for stringology problems (e.g. \cite{ChanGKKP20, MASEK198018, MyersRegular, Indyk98a, CrochemoreLZ03, BilleF08, BilleThorup, Grabowski16}). 
 In this paper, we will tackle the following decades-old question:

\begin{open}
\label{open:abrahamson}
Can we improve Abrahamson's $O(n\sqrt{m\log m})$ time algorithm for Text-to-Pattern Hamming Distances?
\end{open}

 We remark that Fredriksson and Grabowski \cite{FredrikssonG13} designed a faster algorithm for Text-to-Pattern Hamming Distances  when the word size $w$ is $\omega(\log n)$ and $m = O(\frac{n \log m}{w})$. However, their algorithm is not faster in the common Word-RAM model with $w = \Theta(\log n)$, which is the model we consider here.

To obtain faster algorithms for the Text-to-Pattern Hamming Distances problem, researchers have considered two easier versions: the $(1+\eps)$-approximate version and the $k$-bounded version. Next we summarize previous results in these two settings.

\paragraph*{$(1+\eps)$-approximation.}
 The $(1+\eps)$-approximate version asks to approximate the Hamming distance  between $P$ and $T[i \mathinner{.\,.} i + m-1]$ for every shift $i$ within a $(1 + \eps)$ factor of the true distance, for $\eps>0$.

In 1993, Karloff \cite{Karloff93} gave a randomized (Monte Carlo) algorithm running in $O(\eps^{-2}n\log n\log m)$ time with high success probability.\footnote{With high probability (w.h.p.) means probability $1-O(n^{-c})$ for arbitrarily large constant $c$.} Karloff also derandomized his algorithm at the cost of only an extra logarithmic factor.

Karloff's $\widetilde O(\eps^{-2}n)$ time algorithm remained the state-of-the-art for a long time (and unimproved except in some special cases~\cite{AtallahGW13}).\footnote{We use $\OO$ to hide poly-logarithmic factors in the input size.} 
This $\eps^{-2}$ dependency is due to the variance that results from random projections, and it was thought to be inherent as suggested by the $\Omega(\eps^{-2})$ lower bound for computing Hamming distance in the one-way communication model \cite{Woodruff04,JayramKS08}. Hence, it came as a surprise when Kopelowitz and Porat in STOC'15 \cite{KopelowitzP15} gave a faster algorithm in $\widetilde O(\eps^{-1}n)$ time, using techniques from sparse recovery.
This algorithm was subsequently simplified (and improved in terms of logarithmic factors) by Kopelowitz and Porat \cite{KopelowitzP18}, with a time bound of $O(\eps^{-1}n\log n\log m)$ (with high success probability). See Table~\ref{tbl:approx}. 
 It is unclear whether this $\eps^{-1}$ dependency is best possible, and this leads to the following tantalizing question:

\begin{open}
\label{open:approx}
Can we improve Kopelowitz and Porat's $\widetilde O(\eps^{-1}n)$ time algorithm for $(1+\eps)$-approximate Text-to-Pattern Hamming Distances?
\end{open}

Generally, there has been growing interest in understanding the $\eps$-dependencies needed to solve fundamental problems in fine-grained complexity
(e.g., partition~\cite{MuchaW019,BringmannN21,linchen} and knapsack~\cite{Chan18a,Jin19,BringmannC22, DengJM23,xiaomao,guochuanzhang}).
Such $\eps$-dependencies are especially important when one demands very accurate answers (e.g., computing $(1+\frac{1}{\sqrt{m}})$-approximations).

More recently, Chan, Golan, Kociumaka, Kopelowitz and Porat in STOC'20 \cite{ChanGKKP20} partially answered \cref{open:approx}:  when the pattern length $m$ satisfies $m\ge \eps^{-28}$, one can $(1+\eps)$-approximate Text-to-Pattern Hamming Distances in $\widetilde O(n)$ time, without any $\eps^{-O(1)}$ factors.
The assumption may be relaxed to $m\ge \eps^{-10}$ if
the matrix multiplication exponent $\omega$ is equal to 2, and if the goal is to obtain better than $\OO(\eps^{-1}n)$ time instead of $\OO(n)$, the assumption can be relaxed further by re-analyzing/modifying their algorithm.
However, inherently their approach is unable to beat $\eps^{-1}n$ if $\eps^{-1}$ is large, for example, when $\eps^{-1}$ is $m^{1/3}$ or $\sqrt{m}$.
The $\eps^{-1}=\sqrt{m}$ case is particularly instructive: here, $\OO(\eps^{-1}n)$ coincides with the $\OO(n\sqrt{m})$  bound for the exact problem; for distances that are $\Theta(m)$, we are demanding $O(\sqrt{m})$ additive error, and sampling-based approaches do not seem to offer any speedup (if we try to estimate distances by sampling different positions of the pattern string, we would need a sample size of $\Omega(m)$, which is not any smaller than the length of the original pattern string).

Chan et al.~\cite{ChanGKKP20} also gave an $O(\eps^{-2}n)$ time randomized algorithm (correct with high probability)    without any logarithmic factors, which is preferable when  $\eps^{-1}$ is small. Both \cref{open:abrahamson} and \cref{open:approx} were explicitly asked during a talk on \cite{ChanGKKP20} given by Kociumaka.\footnote{Talk video link: \url{https://youtu.be/WEiQjjTBX-4?t=2820}}

Other variations of $(1+\eps)$-approximation text-to-pattern matching problem have also been studied in the literature, such as replacing Hamming distance by other $\ell_p$ norms 
\cite{LipskyP11,GawrychowskiU18,StudenyU19,Uznanski20approx}, or restricting to algorithms that do not use FFT \cite{ChanGKKP20,Uznanski20approx}. 
See also the survey by Uzna\'{n}ski~\cite{Uznanski20recent}.

\paragraph*{$k$-mismatch.}
 Given a threshold $k$, the $k$-bounded (or $k$-mismatch) version of Text-to-Pattern Hamming Distances asks to compute the Hamming distances only for locations
with distances at most $k$, and output $\infty$ for other locations. 

After a long line of works \cite{LandauV86,LandauV89,GalilG86,SahinalpV96,ColeH02,AmirLP04,CliffordFPSS16,GawrychowskiU18,ChanGKKP20}  
(see Table~\ref{tbl:exact}), the current fastest algorithm by Chan, Golan, Kociumaka, Kopelowitz and Porat \cite{ChanGKKP20} is a Monte Carlo randomized algorithm (correct with high  probability) in $O\big (n + \min \big (\frac{nk}{\sqrt{m}}\sqrt{\log m},\frac{nk^2}{m}\big )\big )$ time, shaving off some logarithmic factors from the earlier deterministic algorithm by Gawrychowski and Uzna\'{n}ski \cite{GawrychowskiU18} in $O\big ( n \log ^2 m\log \sigma + \frac{nk\sqrt{\log n}}{\sqrt{m}} \big )$ time. Gawrychowski and Uzna\'{n}ski \cite{GawrychowskiU18} also extended Indyk's fine-grained reduction (mentioned in the notes of Clifford \cite{indyk-lb}) to show a tight conditional lower bound for combinatorial algorithms solving the $k$-mismatch problem under the Boolean Matrix Multiplication Hypothesis.

\subsection{Our results}
In this paper, we give new exact and approximation algorithms for Text-to-Pattern Hamming Distances, answering both \cref{open:abrahamson} and \cref{open:approx} in the affirmative.

\begin{restatable}[Approximation algorithm with sublinear $1/\eps$ dependence]{theorem}{approxhammingmain}
\label{thm:approxhammingmain}
The $(1+\eps)$-approximate Text-to-Pattern Hamming Distances problem can be solved by
a Monte Carlo randomized algorithm in $\widetilde O(\eps^{-\gamma} n)$ time, where $\gamma = 8/(11-\omega) < 0.928$.
\end{restatable}
Here, $\omega\in [2,2.372)$ is the matrix multiplication exponent \cite{DWZ22, journals/corr/abs-2307-07970}.
Our result resolves \cref{open:approx}, showing that $\OO(\eps^{-1}n)$ \cite{KopelowitzP15,KopelowitzP18} is not the ultimate answer for this problem (and, in particular, that it is possible to obtain polynomial improvements over $\OO(n\sqrt{m})$ in the critical case of $\eps^{-1}=\sqrt{m}$).

\begin{restatable}[Exact algorithm without log factors]{theorem}{exacthammingmain}
\label{thm:exacthammingmain}
The Text-to-Pattern Hamming Distances problem can be solved by
a Las Vegas algorithm  which terminates in $O(n\sqrt{m})$ time with high probability. 
\end{restatable}
This result is the first speedup over Abrahamson's  algorithm \cite{Abrahamson87} for more than three decades. 
We also give a new deterministic algorithm that runs faster than Abrahamson's algorithm, but slower than \cref{thm:exacthammingmain}.

\begin{restatable}[Deterministic exact algorithm]{theorem}{exacthammingdetmain}
\label{thm:exacthammingdetmain}
The Text-to-Pattern Hamming Distances problem can be solved by
a deterministic algorithm in $O(n\sqrt{m} (\log m\log \log m)^{1/4})$ time.
\end{restatable}

Our exact algorithms actually apply to the harder problem of computing $|\{j: P[j]\le T[i+j]\}|$ for all shifts $i$. (Note that the Text-to-Pattern Hamming Distances problem reduces to two instances of this problem, one of which with an alphabet in reversed order.) This is also known as the \emph{Dominance Convolution} problem (see e.g., \cite{AmirF91,LabibUW19}).

\begin{restatable}[Exact algorithm for Text-to-Pattern Dominance Matching]{theorem}{exactdominance}
\label{thm:exactdominancemain}
The Text-to-Pattern Dominance Matching problem can be solved by
a Las Vegas algorithm 
which terminates in $O(n\sqrt{m})$ time with high probability, or a deterministic algorithm which terminates in $O(n\sqrt{m} (\log m\log \log m)^{1/4})$ time.
\end{restatable}

\cite{LabibUW19} also observed that this Text-to-Pattern Dominance Matching problem is equivalent to the following ``threshold'' problem~\cite{AtallahD11}: for a fixed $\delta$, compute $|\{j: |P[j]- T[i+j]|>\delta\}|$ for all shifts $i$. Hence, this threshold problem can also be solved in the same time complexity as in \cref{thm:exactdominancemain}.

Our technique also yields improvement to the $k$-mismatch problem.

\begin{restatable}[$k$-mismatch algorithm without log factors]{theorem}{kmismatchmain}
\label{thm:kmismatchmain}
The $k$-bounded Text-to-Pattern Hamming Distances problem can be solved by
a Monte Carlo algorithm in 
$O\big (n + \frac{nk}{\sqrt{m}}\big )$
 expected time which outputs correct answers with high probability.
\end{restatable}
This speeds up the previous $O\big (n + \min \big (\frac{nk}{\sqrt{m}}\sqrt{\log m},\frac{nk^2}{m}\big )\big )$-time Monte Carlo  algorithm (with high success probability) 
 time by  Chan, Golan, Kociumaka, Kopelowitz and Porat  \cite{ChanGKKP20},
and cleans up \emph{all} the extra factors from the long line of previous works shown in Table~\ref{tbl:exact} (note that 
$n+\frac{nk^2}{m}$ is never better than $n+\frac{nk}{\sqrt{m}}$).

Finally, we consider the fine-grained complexity of Text-to-Pattern Hamming Distances. As mentioned earlier, a reduction by Indyk (see \cite{indyk-lb}) gives a tight conditional lower bound for combinatorial  Text-to-Pattern Hamming Distances algorithms under the Boolean Matrix Multiplication Hypothesis. Indyk's reduction only gives an $n\cdot m^{\omega/2-1-o(1)}$ conditional lower bound for arbitrary (potentially non-combinatorial) algorithms. This lower bound is only non-trivial if $\omega>2$.

In the Appendix we observe\footnote{The authors thank Amir Abboud and Arturs Backurs for discussions around this observation in 2015--2016.}
that Indyk's reduction can be easily extended to start from the Equality Product of matrices \cite{vnotes}, which is known to be equivalent to  Dominance Product \cite{VassilevskaW09,MatIPL,LabibUW19,vnotes}). Equality Product and Dominance Product are among the so called ``intermediate'' matrix products \cite{Lincoln0W20} believed to require $n^{2.5-o(1)}$ time, even if $\omega=2$ (see also \cite{LabibUW19}). The observation gives a higher, $nm^{1/4-o(1)}$ time fine-grained lower bound for Text-to-Pattern Hamming Distances against potentially non-combinatorial algorithms which holds even if $\omega=2$. Similarly, Gawrychowski and Uzna\'{n}ski's  reduction \cite{GawrychowskiU18} from Matrix Multiplication to the $k$-mismatch problem can also be extended this way, giving a higher $\frac{n^{1-o(1)}k}{m^{3/4}}$ fine-grained lower bound against potentially non-combinatorial algorithms which holds even if $\omega=2$ and is only off by an $m^{1/4}$ factor  from the known combinatorial algorithms for the problem. 

Finally, we examine the relationship between Text-to-Pattern Hamming Distances and the well-studied $3$SUM problem. It has long been asked (see e.g. \cite{Uznanski20recent}) whether one can reduce $3$SUM to Text-to-Pattern Hamming Distances.

Recently, Chan, Vassilevska Williams and Xu \cite{ChanVX23} showed that $3$SUM is {\em fine-grained equivalent} to the following counting version called $\#$All-Nums-$3$SUM.

\begin{prob}[$\#$All-Nums-$3$SUM]
\label{prob:allnumscount}
Given three size $N$ sets $A, B, C$ of integers, for every $c \in C$, compute the number of $(a, b) \in A \times B$ where $a + b =c$. 
\end{prob}

We consider the following variant of $\#$All-Nums-$3$SUM in which one of the input sets is assumed to contain integers from a small range ($3$SUM  where the numbers of one of the three sets are from a small range was  mentioned in~\cite{ChanL15}).

\begin{prob}
\label{prob:$3$SUM_variant}
Given three size $N$ sets $A, B, C$ where $C = [N]$, for every $c \in C$, compute the number of $(a, b) \in A \times B$ where $a + b =c$. 
\end{prob}

We show that Text-to-Pattern Hamming Distances when $n = O(m)$ is {\em equivalent} to Problem \ref{prob:$3$SUM_variant}. This at least partially addresses the relationship between Text-to-Pattern Hamming Distances and $3$SUM, as Problem \ref{prob:$3$SUM_variant} can be viewed as a range-restricted version of $3$SUM (as $3$SUM is equivalent to $\#$All-Nums-$3$SUM).

\begin{restatable}[Equivalence with a variant of $3$SUM]{theorem}{equivalence}
\label{thm:equivalence}
If \cref{prob:$3$SUM_variant} has a $f(N)$ time algorithm, then Text-to-Pattern Hamming Distances with $n = O(m)$ has an $\OO(f(m))$ time algorithm, and vice versa. 
\end{restatable}

Bringmann and Nakos \cite{BringmannN20} designed a reduction from  Text-to-Pattern Hamming Distances to a problem called Interval-Restricted Convolution, which is more general than \cref{prob:$3$SUM_variant}, and their reduction also works from  Text-to-Pattern Hamming Distances to \cref{prob:$3$SUM_variant}. We show that the reduction is also possible in the other direction from \cref{prob:$3$SUM_variant} to  Text-to-Pattern Hamming Distances, establishing the equivalence.

\renewcommand{\arraystretch}{1.3}

\begin{table}\centering\small
\begin{tabular}{|l|l|}\hline
reference & run time\\\hline\hline
Fischer and Paterson~\cite{fispat} &  $O(\sigma n\log m)$\\
Abrahamson~\cite{Abrahamson87} & $O(n\sqrt{m\log m})$\\
new & $O(n\sqrt{m})$\\\hline
Landau and Vishkin~\cite{LandauV86,LandauV89} / Galil and Giancarlo~\cite{GalilG86} & $O(nk)$\\
Sahinalp and Vishkin~\cite{SahinalpV96}  & $O(n + \frac{nk^{O(1)}}{m})$\\
Cole and Hariharan~\cite{ColeH02} & $O(n + \frac{nk^4}{m})$\\
Amir, Lewenstein and Porat~\cite{AmirLP04} & $O(\min\{n\sqrt{k\log k},\, 
n\log k + \frac{nk^3\log k}{m}\})$\\
Clifford, Fontaine, Porat, Sach, and Starikovskaya~\cite{CliffordFPSS16} & $O(n\log^{O(1)}m + \frac{nk^2\log k}{m})$\\
Gawrychowski and Uzna\'{n}ski~\cite{GawrychowskiU18} & $O(n\log^2 m\log\sigma +
\frac{nk\sqrt{\log n}}{\sqrt{m}})$\\
Chan, Golan, Kociumaka, Kopelowitz and Porat~\cite{ChanGKKP20} &
$O(n + \min\{\frac{nk\sqrt{\log m}}{\sqrt{m}},\, \frac{nk^2}{m}\})$\\
new & $O(n + \frac{nk}{\sqrt{m}})$\\\hline
\end{tabular}
\caption{Exact algorithms for the text-to-pattern Hamming distance problem (randomization allowed).}
\label{tbl:exact}
\end{table}

\begin{table}[ht]\centering\small
\begin{tabular}{|l|l|l|}\hline
reference & run time & techniques\\\hline
Karloff~\cite{Karloff93} & $\OO(\eps^{-2}n)$ & random projection\\ 
Indyk~\cite{Indyk98a} & $\OO(\eps^{-3}n)$ & random sampling\\
Kopelowitz and Porat~\cite{KopelowitzP15} &
$\OO(\eps^{-1}n)$ & projection with sparse recovery\\
Kopelowitz and Porat~\cite{KopelowitzP18} &
$\OO(\eps^{-1}n)$ & random projection\\
Chan, Golan, Kociumaka, Kopelowitz and Porat~\cite{ChanGKKP20} &
$\OO(n)$ for $m\gg \eps^{-28}$ & random sampling $+$ rect.\ matrix mult.\\
new & $\OO(\eps^{-0.93}n)$ & \#3SUM techniques with matrix mult.\\\hline
\end{tabular}
\caption{Approximation algorithms for the text-to-pattern Hamming distance problem, focusing on $\eps$-dependencies and ignoring logarithmic factors (randomization allowed).}
\label{tbl:approx}
\end{table}

\subsection{Technical overview}

Both our approximation and exact algorithms for
Text-to-Pattern Hamming Distances
use interesting new techniques, on which we now briefly elaborate.

\paragraph{Approximation algorithm.}
Our new approximation algorithm for Theorem~\ref{thm:approxhammingmain} uses an approach that is markedly different from all 
previous approximation algorithms.
The algorithms by Karloff~\cite{Karloff93} and Kopelowitz and Porat~\cite{KopelowitzP18} used random projection to reduce the alphabet size; afterwards, the problem can be solved by FFT\@.  Karloff's algorithm
required $O(\eps^{-2})$ projections, whereas Kopelowitz and Porat's
algorithm required a reduced alphabet size of $O(\eps^{-1})$.  On the other hand, the algorithms by Indyk~\cite{Indyk98a} and Chan et al.~\cite{ChanGKKP20} used
random sampling to examine selected positions of the pattern and text strings.  An application of the Chernoff bound 
leads to $O(\eps^{-2})$ factors in the running time, which are too big for
the critical case $\eps^{-1}=\sqrt{m}$.

In contrast, our new algorithm follows an approach that is actually 
closer to the known exact algorithms.  We view the problem as 
a certain colored counting variant of $3$SUM (where colors correspond to characters
in the alphabet), which can be decomposed into multiple instances of an
uncolored counting $3$SUM problem (one per character in the alphabet).

Recently, Chan, Vassilevska Williams and Xu~\cite{ChanVX23} gave reductions from
counting versions of basic problems in fine-grained complexity,
including $3$SUM and Exact-Triangle (finding triangles with weight exactly zero in a dense weighted graph), to their original versions.
They obtained their results using a simple combination of
``Fredman's trick''\footnote{The (trivial) observation that $a+b\le a'+b'$ is equivalent to $a-a'\le b'-b$, which is the key insight behind Fredman's seminal paper on all-pairs shortest paths (APSP)~\cite{fredman1976new}, and used in many subsequent works on APSP and 3SUM (e.g., \cite{Takaoka98,Chan10,Williams18,GronlundP18, chan3sum}).} and  Equality Product.
We show that their ideas, originally developed for proving fine-grained equivalences and conditional lower bounds, can be adapted to design faster algorithms for \emph{approximate} counting versions of $3$SUM and
Exact-Triangle.

More specifically, Fredman's trick and Equality Product allow us to compute counts in the case when counts are large (i.e., when the number of ``witnesses'' is large)~\cite{ChanVX23}.  Chan, Vassilevska Williams and Xu then
used oracles to handle the case when counts are small (since they were designing reductions, not algorithms), but we observe that
the small-count case is actually easier in the context 
of approximation: we can use random sampling with smaller sample sizes, since the variance is lower. 

To summarize, our new algorithm is technically interesting for
multiple reasons:

\begin{enumerate}
\item It further illustrates the power of the ``Fredman's trick meets Equality Product'' technique from \cite{ChanVX23}, in the context
of approximation algorithms.  These ideas might spawn further applications.
\item It demonstrates that fine-grained reductions 
(originally developed for proving conditional lower bounds) can
help in the design of algorithms.  In particular, our algorithm makes essential use of
the known chain of reductions from  Convolution-$3$SUM to Exact-Triangle ~\cite{VassilevskaW09}, and from $3$SUM to
Convolution-$3$SUM~\cite{ChanHe}.
\item Its use of matrix multiplication is non-trivial and interesting.
It is open whether fast matrix multiplication helps for the exact Text-to-Pattern Hamming Distances problem, but our new algorithm demonstrates that it helps for the approximate problem.  Chan et al.'s previous
algorithm~\cite{ChanGKKP20} also used rectangular matrix multiplication
to speed up certain steps, but our algorithm here relies on matrix multiplication (via Equality Product) in a more essential way.
\end{enumerate}

\paragraph{Exact algorithms.}
Our exact algorithm for Theorem~\ref{thm:exacthammingmain}
also works by decomposing into multiple $3$SUM-like subproblems (one per character of the alphabet). More precisely, a subproblem corresponds to
computing a sumset $X+Y$ (where we also want the counts/multiplicities per element) for
two given sets of $n$ integers in $[U]$; equivalently, this corresponds
to computing the convolution of two sparse binary vectors in $[U]$ with $n$ nonzero entries.
The critical case in our application turns out to be when $U$ is below and close to $n^2$;
this is when 
the standard  $O(U\log U)$-time algorithm by FFT does not outperform  the brute-force $O(n^2)$-time algorithm.
We present a new lemma showing that the sumset/convolution can be
computed in $O(U\log (n^2/U))$ expected time, which outperforms both FFT and brute-force, and is good enough to shave
off all the extra logarithmic factors and yield the $O(n\sqrt{m})$ bound
for the exact Text-to-Pattern Hamming Distances problem.

A number of algorithms have already been developed for sparse convolution
\cite{ColeH02,ChanL15,BringmannFN22}.  The current fastest algorithm by Bringmann et al.~\cite{BringmannFN22}
is complicated, and does not address our question of speeding up
$O(U\log U)$ in the regime of $U$ close to $n^2$.
Our new algorithm shares some ideas from these previous algorithms, but
is arguably simpler, and more accessible, than the one of \cite{BringmannFN22}.  It only requires
Dietzfelbinger's standard family of almost linear hash functions~\cite{Dietzfelbinger96} (though we need to establish some new properties).
Hashing is used to iteratively shrink the ``support'' (the subset of elements whose counts are not known yet), but the key twist is an extra use
of recursion to identify candidates for ``light'' elements (elements with small counts).  A bit-packed version of FFT  is used to compute counts with a small modulus.

By further incorporating some ideas from \cite{ChanL15}, we obtain 
a derandomization (Theorem~\ref{thm:exacthammingdetmain}), albeit with an extra factor of about $\log^{1/4}m$.

It is straightforward to combine our new lemma with existing algorithms
\cite{GawrychowskiU18,ChanGKKP20} to obtain our result on the $k$-mismatch problem (Theorem~\ref{thm:kmismatchmain}) and the Dominance Convolution problem (Theorem~\ref{thm:exactdominancemain}).

\section{Preliminaries}

A string $S$ of length $|S| = s$ is a sequence of characters $S[1]S[2]\dots   S[s]$ over an alphabet $\Sigma$. We assume the alphabet has size $|\Sigma| = \sigma \le n^{O(1)}$, and identify $\Sigma$ with the set of integers in $[\sigma]$.
For $1\le i\le j\le s$, we denote the substring $S[i]S[i + 1] \dots S[j]$ of $S$ by $S[i \dd j]$.

We use $[n]$ to denote $\{0, 1, \ldots, n-1\}$. 

\begin{definition}
Given two length-$n$ sequences $\langle a_0,\ldots,a_{n-1}\rangle$ and $\langle b_0,\ldots,b_{n-1}\rangle$, their convolution $c = a \star b$ is a length $2n-1$ sequence, where 
$c_i = \sum_{j=0}^{i} a_i b_{j-i}$ (assume that out-of-range
array entries are set to $0$). 
\end{definition}

It is well-known that we can compute the convolution between two integer sequences in $O(n \log n)$ time using FFT. If one instead only needs to compute the entries of $c$ modulo some given prime $p$, then a slightly faster running time is possible, in the word-RAM model with $O(\log n)$ bit words:
\begin{restatable}{lemma}{packedFFT}
\label{lem:packed_FFT}
Given a prime $p \le n^{O(1)}$ and two length-$n$ sequences $a, b$ with entries in $\mathbb{F}_p$,  we can deterministically compute $a \star b$ in $O(n \log p)$ time. 
\end{restatable}
Indyk \cite{Indyk98a} claimed a proof of
\cref{lem:packed_FFT} for the case of $p=2$, but his proof was incomplete.
For the $p=2$ case, his argument can be completed via a more recent work \cite{LinAHC16}, but it does not extend to larger $p$. 
In \cref{app:poly-mult} we point out the issue in Indyk's argument (and mention subsequent works affected by this issue),  and then include a complete proof of \cref{lem:packed_FFT} for general $p$ that fixes this issue.

We use $M(n_1,n_2,n_3)$ to denote the time for computing the product between an $n_1 \times n_2$ matrix and an $n_2 \times n_3$ matrix. We use the following algorithm for computing the equality product between two matrices, which follows from known techniques \cite{MatIPL, YusterDom} (see also \cite{ChanVX23}).  

\begin{lemma}\label{lem:eqprod}
Given an $n_1\times n_2$ matrix $A$ and an $n_2\times n_3$ matrix $B$, 
their \emph{equality product} 
\[ C[i,j] := |\{k\in [n_2]: A[i,k]=B[k,j]\}| \]
can be computed in time
\[ \Meq(n_1,n_2,n_3)\ =\ O\left( \min_{1\le r\le n_2}\left(\frac{n_1n_2n_3}{r} +
M(n_1,rn_2,n_3)\right)\right).
\]
\end{lemma}

For example, in the case of square matrices, the above implies
$\Meq(n,n,n)=O(\min_r (n^3/r + M(n,rn,$ $n)))\le O(\min_r (n^3/r + rn^\omega))=
O(n^{(3+\omega)/2})$.

\section{Approximate Text-to-Pattern Hamming Distances}

In this section, we begin by solving approximate counting variants of
several core problems in fine-grained complexity---namely, Exact Triangle, Convolution-$3$SUM, and $3$SUM\@.  All this will then lead to an algorithm for
approximate Text-to-Pattern Hamming Distances.

\subsection{Approximate Counting All-Edges Exact Triangle}

In recent work \cite{ChanVX23}, Chan, Vassilevska Williams and Xu proved fine-grained equivalences between several central fine-grained problems and their counting versions. Let us take the All-Edges Exact Triangles (AE-Exact-Triangle) problem for an example. In this problem, when given a weighted tripartite graph, and for each edge, we need to decide whether this edge is in a triangle whose edge weights sum up to $0$. In the counting version of AE-Exact-Triangle, we need to count the number of zero-weight triangles each edge is in. 
Equivalently, given 3 matrices $A$, $B$, and $C$, we need to count the number of $k$s such that $A[i,k]+B[k,j]=-C[i,j]$, for each $i,j$.
Prior to \cite{ChanVX23}, via the technique in~\cite{focsyj}, it was known that AE-Exact-Triangle is equivalent to its counting version where all the per-edge counts are small. The key observation in \cite{ChanVX23} is that, when the per-edge counts are big, we can efficiently compute the counts exactly, using Fredman's trick \cite{fredman1976new} in combination with Equality Product. 

The following lemma adapts this approach to an approximate counting setting with additive error.  The main new idea is that when counts are small, we can use random sampling at a lower rate to estimate such counts, since the variance is lower.  In fact, even in the case when counts are big, we can also use sampling at different rates to approximate the Equality Products a little more quickly.

Let us briefly discuss the lemma statement below.
First, in the approximate setting of AE-Exact-Triangle, it turns out that
the third matrix $C$ is unnecessary: one can design a truly subcubic time algorithm
 that solves the problem for all $C$ at the same time.
Intuitively, the reason is that when additive error is allowed, small counts in principle may be approximated by zeros, and zero values need not be output explicitly. Thus, it suffices to estimate the count values $\COUNT[i,j,z]:=|\{k\in [n_2]: A[i,k]+B[k,j]=z\}|$ for the ``heavy hitters'' $z$ with sufficiently large
counts, but the number of such $z$ is sublinear per $(i,j)$.  

Second, the lemma statement below bounds the variance of the estimators, instead of the additive error (which is bounded by the square root of the variance with good probability); this will be important, as we will later need to sum several estimators 
(if they are independent or uncorrelated, we can sum their variances).

Third, the lemma involves multiple parameters, and on first reading,
it may be helpful to focus, throughout this section, on the simplest setting with $t=1$ and $\Delta=1$ (where
one can ignore the condition about uncorrelation), which is already sufficient to address the critical case of the Hamming Distances problem when
$\eps^{-1}=\sqrt{m}$ 
and distances are $\Theta(m)$ (where we want $O(\sqrt{m})$ additive error).

\begin{lemma}\label{lem:approx:exacttri}
Given an $n_1\times n_2$ integer matrix $A$ and an $n_2\times n_3$ integer matrix $B$,
where all values of $A$ are divisible by positive integer $\D \le n_3$, 
define 
\[ \COUNT[i,j,z]:=|\{k\in [n_2]: A[i,k]+B[k,j]=z\}|. \]
Given parameter $1 \le t \le n_2$, there is a randomized algorithm
that computes estimates $f[i,j,z]$ over all $i\in [n_1]$, $j\in [n_3]$, and $z$, such 
that the expectation of $f[i,j,z]$ is equal to $\COUNT[i,j,z]$,
and the variance of $f[i,j,z]$ is $O(t n_2)$.  
(Zero entries of $f$ need not be output explicitly.)
Furthermore, $f[i,j,z]$ and $f[i',j',z']$ are uncorrelated if $(i,j)\neq (i',j')$ and
$z\not\equiv z'\!\!\pmod{\D}$.
The running time is
\[ \OO\left( \min_{1 \le s \le n_2 / t}\left( \frac{n_1n_2n_3}{st} + 
s\D M(n_1,n_2/t,n_3/\D)\right) \right).
\]
\end{lemma}

\begin{proof}
Define the \emph{witness set}  $W[i,j,z] = \{k\in [n_2]: A[i,k]+B[k,j]=z\}$.

\begin{itemize}
\item {\bf Few-witnesses case.}
Independently for each $(i, j)$, pick a random subset $R_{ij}\subseteq [n_2]$ where each element in $[n_2]$ is put in $R$
with probability $\rho_* := 1/(st)$.
Define $f_*[i,j,z] = (1/\rho_*)\cdot |\{k\in R_{ij}: A[i,k]+B[k,j]=z\}|$.
We can generate all nonzero values $f_*[i,j,z]$ by looping through
each $i\in [n_1]$, $j\in [n_3]$, and $k\in R_{ij}$, in total time
$\OO(n_1n_3\cdot \rho_* n_2)$.  (It is essential that we are not required to output zero values, since otherwise we would not be able to obtain $o(n_1n_2n_3)$ time.)
Then $\Ex[f_*[i,j,z]]=\COUNT[i,j,z]$,
and $\Var[f_*[i,j,z]]=(1/\rho_*^2) \cdot \COUNT[i,j,z]\cdot (\rho_*-\rho_*^2) \le st\cdot \COUNT[i,j,z]$, which is good when $\COUNT[i,j,z]$ is small. As we pick $R_{ij}$ independently for each $(i, j)$, the estimates $f_*[i, j, z]$ and $f_*[i', j', z']$ are independent if $(i, j) \ne (i', j')$.

\item {\bf Many-witnesses case.}
For $\ell=0,\ldots,\log s$ (w.l.o.g., we assume $s$ to be a power of $2$), do the following:

Pick a random subset $H^{(\ell)}\subseteq [n_2]$ of size $c2^\ell\log(n_1n_2n_3)$ for a sufficiently large constant $c$.  
Then $H^{(\ell)}$ hits $W[i,j,z]$ w.h.p.\ for each
$i,j,z$ with $\COUNT[i,j,z]\ge n_2/2^{\ell}$.

For each $k_0\in H^{(\ell)}$ and $\xi\in[\D]$,
pick a random subset $R^{(\ell,k_0,\xi)}\subseteq [n_2]$ where each element in $[n_2]$ is put in $R^{(\ell,k_0,\xi)}$
with probability $\rho_\ell := 1/(2^\ell t)$.
For each $k_0\in H^{(\ell)}$ and $\xi\in[\D]$,
compute the equality product
$C^{(k_0,\ell)}[i,j] = |\{k\in R^{(\ell,k_0,\xi)}: A[i,k]-A[i,k_0] = B[k_0,j]-B[k,j]\}|$ for $i\in [n_1]$ and $j\in [n_3]$ with $B[k_0,j]\bmod \D =\xi$.  
This takes total time $\OO(2^\ell\D\Meq(n_1,\rho_\ell n_2,n_3/\D))$, by splitting each set $\{j\in [n_3]: B[k_0,j]\bmod\D = \xi\}$ into subsets
of size $O(n_3/\D)$.

Consider $i,j,z$ such that $W[i,j,z]$ is hit by $H^{(\ell)}$, i.e.,
$z=A[i,k_0]+B[k_0,j]$ for some $k_0\in H^{(\ell)}$.  Let $k_0$ be the smallest
such index.
Define $f_\ell[i,j,z] = (1/\rho_\ell)\cdot C^{(k_0,\ell)}[i,j]$.
Then $\Ex[f_\ell[i,j,z] \mid H^{(\ell)}] = |\{k\in [n_2]: A[i,k]-A[i,k_0]=B[k_0,j]-B[k,j]\}|
= |\{k\in [n_2]: A[i,k]+B[k,j]=A[i,k_0]+B[k_0,j]=z\}|
= \COUNT[i,j,z]$ by Fredman's trick.
And $\Var[f_\ell[i,j,z] \mid H^{(\ell)}] = (1/\rho_\ell^2)\cdot \COUNT[i,j,z]\cdot (\rho_\ell-\rho_\ell^2)\le 2^\ell t \cdot \COUNT[i,j,z]$.

Note that $f_\ell[i,j,z]$ (conditioned on a fixed $H^{(\ell)}$) depends on $R^{(\ell,k_0,\xi)}$ for
$\xi = B[k_0,j]\bmod \D = z\bmod \D$.  Thus, if $z\not\equiv z'\!\!\pmod \D$, then $f_\ell[i,j,z]$ and
$f_{\ell}[i',j',z']$ are independent conditioned on any fixed $H^{(\ell)}$. 
\end{itemize}

Finally, define $f[i,j,z]$ to be $f_\ell[i,j,z]$ for the smallest index
$\ell\in [\log s]$ such that $W[i,j,z]$ is hit by $H^{(\ell)}$, or $f_*[i,j,z]$ if $\ell$ does not exist.  Then $\Ex[f[i,j,z] \mid \{H^{(\ell)}\}_\ell]=\COUNT[i,j,z]$ for any fixed $\{H^{(\ell)}\}_\ell$, which implies $\Ex[f[i,j,z]]=\COUNT[i,j,z]$. 

If $\COUNT[i,j,z]\in [n_2/2^{p+1}, n_2/2^p)$ for $p\in [\log s]$,
then $\ell\le p+1$ w.h.p. Also, $\Var[f[i,j,z] \mid \ell \le p + 1] \le 2^{p+1} t\cdot \COUNT[i,j,z]$ and $\Var[f[i,j,z]$ is polynomially bounded regardless the value of $\ell$, we have that $\Var[f[i,j,z]]\le O(2^{p+1} t\cdot \COUNT[i,j,z])\le O(tn_2)$.  Similarly,
if $\COUNT[i,j,z] < n_2/s$,
then 
$\Var[f[i,j,z]]\le st \cdot \COUNT[i,j,z]\le O(tn_2)$. Note that for $(i,j)\neq (i',j')$ and
$z\not\equiv z'\!\!\pmod{\D}$, $f[i, j, z]$ and $f[i', j', z']$ are independent conditioned on a fixed choice of $\{H^{(\ell)}\}_\ell$. Therefore, $\Ex[f[i, j, z] f[i', j', z'] \mid \{H^{(\ell)}\}_\ell] = \Ex[f[i, j, z] \mid \{H^{(\ell)}\}_\ell] \cdot \Ex[f[i', j', z'] \mid \{H^{(\ell)}\}_\ell] = \COUNT[i, j, z] \cdot \COUNT[i', j', z']$. Summing over all possible $\{H^{(\ell)}\}_\ell$ gives $\Ex[f[i, j, z] f[i', j', z']] = \COUNT[i, j, z] \cdot \COUNT[i', j', z'] = \Ex[f[i,j, z]] \cdot \Ex[f[i',j', z']]$, so $f[i, j, z]$ and $f[i', j', z']$ are uncorrelated. 

The total time is
\begin{eqnarray*}
\lefteqn{\OO\left( \frac{n_1n_2n_3}{st} + \sum_{\ell=0}^{\log s} 2^\ell\D \Meq(n_1,n_2/(2^\ell t),n_3/\D) \right)}\\
&\le& \OO\left( \frac{n_1n_2n_3}{st} +\! \sum_{\ell=0}^{\log s} 
2^\ell\D \left( \frac{n_1n_2n_3}{2^\ell t s\D} +
M(n_1,s n_2/(2^\ell t),n_3/\D) \right)\right) \ \mbox{by Lemma~\ref{lem:eqprod} with $r:=s$}\\
&\le& \OO\left( \frac{n_1n_2n_3}{st} + 
\sum_{\ell=0}^{\log s} 2^\ell\D\cdot (s/2^\ell)\cdot M(n_1,n_2/t,n_3/\D) \right).
\end{eqnarray*}

\vspace{-\bigskipamount}
\end{proof}

\subsection{Approximate Counting All-Numbers Convolution-\texorpdfstring{$3$}{3}SUM}

Next, we apply \cref{lem:approx:exacttri} to Convolution-$3$SUM, by modifying the known reduction \cite{VassilevskaW09} from Convolution-$3$SUM to Exact-Triangle.  In the counting version of the All-Numbers Convolution-$3$SUM problem, we are given 3 sequences $\langle a_0,\ldots,a_{n-1}\rangle$, $\langle b_0,\ldots,b_{n-1}\rangle$, and $\langle c_0,\ldots,c_{n-1}\rangle$, and want to count the number of $k$'s such that $a_k+b_{h-k}=-c_h$, for each $h$.  In the approximate counting version below, again the third sequence is unnecessary.  Note that the time bound below is truly subquadratic ($\OO(n^{(5+\omega)/4})$) for the case $t=\Delta=1$.

\begin{lemma}\label{lem:approx:conv3sum}
Given integer sequences $\langle a_0,\ldots,a_{n-1}\rangle$ and $\langle b_0,\ldots,b_{n-1}\rangle$ where all $a_k$'s are divisible by positive integer $\D \le n$, 
define 
\[ \COUNT[h,z]:=|\{k\in [h]: a_k + b_{h-k}=z\}|. \]
Given parameter $1 \le t \le n$, there is a randomized algorithm
that computes estimates $f[h,z]$ for all $h\in [n]$ and $z$, such 
that the expectation of $f[h,z]$ is equal to $\COUNT[h,z]$,
and the variance of $f[h,z]$ is $O(tn)$.
(Zero entries of $f$ need not be output explicitly.)
Furthermore, $f[h,z]$ and $f[h',z']$ are uncorrelated if $h\neq h'$ and
$z\not\equiv z'\!\!\pmod{\D}$.
The running time is
\[ \OO\left( n^{(5+\omega)/4}/t^{(1+\omega)/4} + n^{(5+\omega)/4}\D^{(3-\omega)/4}/t + n\right).
\]
\end{lemma}
\begin{proof}
Let $d$ be an integer parameter between $1$ and $n$. For simplicity, we assume $n/d$ to be an integer, which can be achieved by padding $\infty$ entries to the arrays. 
We will estimate
$\COUNT^{(\ell)}[id+j,z] := |\{ k\in [d]: a_{(i-\ell)d + k} + b_{\ell d + j- k}=z\}|$
for each $i\in [n/d]$, $j\in [d]$, and $\ell\in [n/d]$.
(Assume that out-of-range
array entries are set to $\infty$.)

Define an $(n/d)\times d$ matrix $A$ and a $d\times d$ matrix $B^{(\ell)}$
for each $\ell\in [n/d]$:
for each $i\in [n/d]$ and $k,j\in [d]$, let
$A[i,k] = a_{id+k}$ and
$B^{(\ell)}[k,j]=b_{\ell d + j-k}$.
(Normally, it would be better to combine into one $d\times n$ matrix $B$
and use rectangular matrix multiplication, but we need multiple independent subproblems here.)
Apply Lemma~\ref{lem:approx:exacttri} to estimate
$\COUNT^{(\ell)}[id+j,z] = |\{k\in [d]: A[i-\ell,k]+B^{(\ell)}[k,j]=z\}|$,
for all $\ell\in [n/d]$, 
in total time
\begin{eqnarray*}
\lefteqn{\OO\left(\frac{n}{d}\cdot\left(\frac{(n/d)\cdot d\cdot d}{st} + s\D M(n/d, d/t, d/\D)\right)\right)}\\
&=& \left\{\begin{array}{ll}
\OO\left(\frac{n^2}{st} + s\D \sqrt{n/t}M(\sqrt{n/t},\sqrt{n/t},\sqrt{tn}/\D) \right) & \mbox{if $t\ge \D$\qquad by setting $d:=\sqrt{tn}$}\\
\OO\left(\frac{n^2}{st} + s\D \sqrt{n/\D}M(\sqrt{n/\D},\sqrt{\D n}/t,\sqrt{n/\D}) \right) & \mbox{if $t<\D$\qquad by setting $d:=\sqrt{\D n}$}
\end{array}\right.\\
&\le& \left\{\begin{array}{ll}
\OO\left(\frac{n^2}{st} + s\D\sqrt{n/t}\cdot (t/\D)\cdot (\sqrt{n/t})^\omega \right) \ =\ \OO\left(\frac{n^2}{st} + \frac{sn^{(\omega+1)/2}}{t^{(\omega-1)/2}}\right)& \mbox{if $t\ge \D$}\\
\OO\left(\frac{n^2}{st} + s\D\sqrt{n/\D}\cdot (\D/t)\cdot (\sqrt{n/\D})^\omega \right) \ =\ \OO\left(\frac{n^2}{st} + \frac{sn^{(\omega+1)/2}\D^{(3-\omega)/2}}{t}\right) & \mbox{if $t< \D$}
\end{array}\right.\\
&=& \left\{\begin{array}{ll}
\OO\left( n^{(5+\omega)/4}/t^{(1+\omega)/4}\right)
& \mbox{if $t\ge \D$\qquad by setting $s:=(n/t)^{(3-\omega)/4}$}\\
\OO\left( n^{(5+\omega)/4}\D^{(3-\omega)/4}/t\right)
& \mbox{if $t < \D$\qquad by setting $s:=(n/\D)^{(3-\omega)/4}$.}\\
\end{array}\right.
\end{eqnarray*}

We can now estimate $\COUNT[h,z] = \sum_{\ell\in [n/d]} \COUNT^{(\ell)}[h,z]$.  The variance of each such estimate is $O((n/d)\cdot td)=O(tn)$, due to independence of the $n/d$ applications of Lemma~\ref{lem:approx:exacttri}.  
\end{proof}

\subsection{Approximate Counting All-Numbers \texorpdfstring{$3$}{3}SUM}

Next, we solve an analogous approximate counting version of $3$SUM, by modify a known reduction from $3$SUM to Convolution-$3$SUM by Chan and He~\cite{ChanHe} (which preserves counts, unlike earlier reductions~\cite{Patrascu10,KopelowitzPP16}).
In the counting version of the All-Numbers $3$SUM problem, we are given 3 sets $A$, $B$, and $C$,
 we want to count the number of $(a,b)\in A\times B$ such that $a+b=-c$, for each $c\in C$. 
 Without the third set $C$, this amounts to computing counts/multiplicities of the elements of the sumset $A+B$.

\begin{lemma}\label{lem:approx:3sum}
Given sets  $A$ and $B$ of $n$ integers, where all elements of $A$ are divisible by $\D \le n$,
define 
\[ \COUNT[z]:=|\{(a,b)\in A\times B: a+b=z\}|. \]
Given parameter $1 \le t \le n$, there is a randomized algorithm
that computes estimates $f[z]$ for all $z$, such 
that the expectation of $f[z]$ is equal to $\COUNT[z]$,
and the variance of $f[z]$ is $\OO(tn)$.
(Zero entries of $f$ need not be output explicitly.)
Furthermore, $f[z]$ and $f[z']$ are uncorrelated if $z\not\equiv z'\!\!\pmod{\D}$.
The expected running time is
\[ \OO\left( n^{(5+\omega)/4}/t^{(1+\omega)/4} + n^{(5+\omega)/4}\D^{(3-\omega)/4}/t + n\right).
\]
\end{lemma}
\begin{proof}
In one of the randomized reductions from $3$SUM to Convolution-$3$SUM in \cite{ChanHe}, one classifies all numbers as bad or good based on the choice of a hash function (it takes expected $O(n)$ time to do this preprocessing). Furthermore, the $3$SUM instance between all the good elements can be reduced to $O(1)$ instances of Convolution-$3$SUM of size $n$, and if the size of the $i$-th set for $i \in [3]$ in the $3$SUM instance is $\frac{n}{k_i}$, then the number of bad elements in it is at most $\frac{n}{2k_i^2}$. The same idea also works in our setting. 

Say we have two sets of size $\frac{n}{k_1}$ and $\frac{n}{k_2}$ respectively, and say it takes $T(\frac{n}{k_1}, \frac{n}{k_2})$ time to solve such an instance. After we apply  \cite{ChanHe}'s hash function in $O(n)$ expected time,  the number of bad elements in the two sets becomes at most $\frac{n}{2k_1^2}$ and $\frac{n}{2k_2^2}$ respectively. For the good elements, we can apply \cref{lem:approx:conv3sum}. We then recursively solve the same problem between all bad elements in the first set, and all elements in the second set, which takes $T(\frac{n}{2k_1^2}, \frac{n}{k_2})$ time. Finally, we solve the same problem between all good elements in the first set and all bad elements in the second set, which takes $T(\frac{n}{k_1}, \frac{n}{2k_2^2})$. Clearly, summing up the results gives the correct expectation. 

The running time can be written as the following recurrence: 
\[ T\left(\frac{n}{k_1},\frac{n}{k_2}\right)
\ \le\ T\left(\frac{n}{2k_1^2}, \frac{n}{k_2}\right) 
+ T\left(\frac{n}{k_1}, \frac{n}{2k_2^2}\right) +
\OO\left(n^{(5+\omega)/4}/t^{(1+\omega)/4} + n^{(5+\omega)/4}\D^{(3-\omega)/4}/t  + n\right).
\]
The recursion tree has depth $O(\log \log n)$, so it has size  $2^{O(\log\log n)} = \log^{O(1)}n$. Therefore, the overall running time is 
$$\OO\left(n^{(5+\omega)/4}/t^{(1+\omega)/4} + n^{(5+\omega)/4}\D^{(3-\omega)/4}/t  + n\right).$$

Say the variance bound for two sets of size $\frac{n}{k_1}$ and $\frac{n}{k_2}$ is $V\left(\frac{n}{k_1},\frac{n}{k_2}\right)$, then we have the following recurrence (conditioned on any fixed hash function): 
\[ V\left(\frac{n}{k_1},\frac{n}{k_2}\right)
\ \le\ V\left(\frac{n}{2k_1^2}, \frac{n}{k_2}\right) 
+ V\left(\frac{n}{k_1}, \frac{n}{2k_2^2}\right) +
O\left(tn\right), 
\]
which can be similarly upper bounded by $\OO(tn)$. 
\end{proof}

\subsection{Approximate Counting Colored All-Numbers \texorpdfstring{$3$}{3}SUM}

\newcommand{\COL}{\mathop{\rm color}}
To solve the Text-to-Pattern Hamming Distances problem, we will actually
need a colored version of counting $3$SUM\@.  This colored problem 
can be solved simply by
independently invoking Lemma~\ref{lem:approx:3sum} for each color class.
Note that  the time bound below is %
sub-$n^{3/2}$  for the case $t=\Delta=1$ and $U=n$.

\begin{lemma}\label{lem:approx:col3sum}
Given sets  $A$ and $B$ of $n$ colored integers in $[U]$, 
where all elements of $A$ are divisible by $\D \le n$,
define 
\[ \COUNT[z]:=|\{(a,b)\in A\times B:\ a+b=z,\ \COL(a)=\COL(b)\}|. \]
Given parameter $1 \le t \le n$, there is a randomized algorithm
that computes estimates $f[z]$ for all $z$, such 
that the expectation of $f[z]$ is equal to $\COUNT[z]$,
and the variance of $f[z]$ is $\OO(tn)$.
Furthermore, $f[z]$ and $f[z']$ are uncorrelated if $z\not\equiv z'\!\!\pmod{\D}$.
The expected running time is
\[ \OO\left(n (U/t)^{(1+\omega)/(5+\omega)} 
 + 
nU^{(1+\omega)/(5+\omega)}\D^{(3-\omega)/(5+\omega)}/t^{4/(5+\omega)} + U\right).
\]
\end{lemma}
\begin{proof}
Let $A_c$ (resp.\ $B_c$) be the subset of all elements of $A$
(resp.\ $B$) of color $c$.  Let $n_c=|A_c|+|B_c|$.
Estimate $\COUNT_c[z]:=|\{(a,b)\in A_c\times B_c: a+b=z\}|$
by Lemma~\ref{lem:approx:3sum} if $n_c\le x$,
or compute it exactly by FFT in $\OO(U)$ time if $n_c > x$.
The number of calls to FFT is $O(n/x)$.
The total time is
\begin{eqnarray*}
\lefteqn{\OO\left( \frac{nU}{x} +  \sum_{n_c\le x} \left( 
n_c^{(5+\omega)/4}/t^{(1+\omega)/4} + n_c^{(5+\omega)/4}\D^{(3-\omega)/4}/t + n_c\right) + U \right)}\\
&=& \OO\left(\frac{nU}{x} + nx^{(1+\omega)/4}/t^{(1+\omega)/4} + nx^{(1+\omega)/4}\D^{(3-\omega)/4}/t + U\right)\\
&=& \OO\left(n(U/t)^{(1+\omega)/(5+\omega)} + 
nU^{(1+\omega)/(5+\omega)}\D^{(3-\omega)/(5+\omega)}/t^{4/(5+\omega)} + 
U\right)
\end{eqnarray*}
by setting $x:=\min\{U^{4/(5+\omega)}t^{(1+\omega)/(5+\omega)},\,
(Ut)^{4/(5+\omega)}/\D^{(3-\omega)/(5+\omega)}\}$.

We can estimate $\COUNT[z]=\sum_c\COUNT_c[z]$, with variance
$\OO(\sum_c tn_c) = \OO(tn)$ due to independence of the
different applications of Lemma~\ref{lem:approx:3sum}.
\end{proof}

The above lemma is sufficient to solve the approximate Text-to-Pattern Hamming Distances problem for distances that are $\Theta(m)$ (where we want additive error $O(\eps m)$), but to deal with more generally distances that are $\Theta(k)$ (with additive error $O(\eps k)$), we need to solve a generalization of the counting colored $3$SUM problem where the input consists of $O(k)$ intervals (i.e., contiguous blocks of integers):

\begin{lemma}
\label{lem:RLE}
Given sets $A$ and $B$ of at most $n$ colored integers in $[n]$, 
define 
\[ \COUNT[z]=|\{(a,b)\in A\times B:\ a+b=z,\ \COL(a)=\COL(b)\}|. \]
Suppose that $A$ and $B$ can be decomposed as unions of $O(k)$ disjoint intervals, where each interval is monochromatic.
There is a randomized algorithm
that computes estimates for $\COUNT[z]$ for all $z$, 
with additive error $\OO(\eps k)$ w.h.p.
The running time is
\[ \OO\left( \eps^{-1+\beta} n^{1-\beta}k^{\beta} +
\eps^{-1-\beta} n^{1+\beta} / k^{2\beta} + n
\right),\qquad
\mbox{where $\beta := (3-\omega)/(5+\omega)$.}
\]
\end{lemma}
\begin{proof}
We may assume that each interval has length at most $n/k$,
by breaking the intervals; the number of intervals remains $O(k)$.
Furthermore, we may assume that each interval is a \emph{dyadic interval},
since each interval can be decomposed into a union of $O(\log n)$ dyadic intervals; the number of intervals remains $\OO(k)$.

We may assume that each interval of $A$ has the same length $L$
and each interval of $B$ has the same length $\ell$, where $L$ and $\ell$
are powers of 2 at most $n/k$, since we can try all pairs $(L,\ell)$ and solve $O(\log^2(n/k))$ instances; the time and additive error bounds increase only by polylogarithmic factors.

W.l.o.g., say $\ell\le L\le n/k$.
Let $A'$ (resp.\ $B')$ denote the  colored set of the $\OO(k)$ left endpoints of the intervals in $A$ (resp.\ $B$).
Estimate $\COUNT'[z]=\{(a,b)\in A'\times B': a+b=z,\ \COL(a)=\COL(b)\}$ by
Lemma~\ref{lem:approx:col3sum}, with variance $\OO(tk)$.  Note that since endpoints in $A'$ and $B'$ are multiples of $\ell$, 
the universe size can
be shrunk to $U := n/\ell$, and  after rescaling, all elements of $A'$ are
divisible by $\D := L/\ell\le n/(k\ell)$.

Let $w(i):= |\{(i_1,i_2)\in [L]\times [\ell]: i_1+i_2=i\}|$,
i.e., $w(i)=\min\{i+1,\, \ell,\, L+\ell-i-1\}$ (for $i\in [L+\ell-1]$).
Then $\COUNT[z] = \sum_{i=0}^{L+\ell-2} w(i)\COUNT'[z-i]$.
From the estimates for $\COUNT'[\cdot]$, we can compute estimates
for $\COUNT[\cdot]$ by doing an FFT in $\OO(n)$ time (or by a more
direct way, since $w(\cdot)$ is just a piecewise-linear function with 3 pieces).
Note that $w(j)\le \ell$ and there are $O(L/\ell)$ nonzero terms $\COUNT'[z-i]$ in the sum; furthermore, after splitting the sum into two halves,
in each half, no two $z-i$ values are equal mod $L$.
Thus, we can estimate the sum of each half, with 
variance
$O((L/\ell)\cdot \ell^2\cdot tk)$, due to pairwise uncorrelation. 
Then we can estimate $\COUNT[z]$ by summing up the two halves, which still has variance 
$O((L/\ell)\cdot \ell^2\cdot tk)$, as $\Var[X+Y] \le 2(\Var[X] + \Var[Y])$ for any random variables $X$ and $Y$.  
This variance bound is $O((\eps k)^2)$
by setting $t:=\eps^2 k/(L\ell)\ge \eps^2k^2/(n\ell)$.
The running time is 
\begin{eqnarray*}
\lefteqn{
\OO\left(k \left(\frac{n/\ell}{\eps^2 k^2/(n\ell)}\right)^{(1+\omega)/(5+\omega)} + 
\frac{  k(n/\ell)^{(1+\omega)/(5+\omega)}(n/(k\ell))^{(3-\omega)/(5+\omega)} }
{ (\eps^2 k^2/(n\ell))^{4/(5+\omega)}}
+ n\right)
}\\
&\le & 
\OO\left( \eps^{-2(1+\omega)/(5+\omega)} n^{2(1+\omega)/(5+\omega)}
k^{(3-\omega)/(5+\omega)} +
\eps^{-8/(5+\omega)} n^{8/(5+\omega)}
/ k^{2(3-\omega)/(5+\omega)} + n
\right).
\end{eqnarray*}

Since the estimate for $\COUNT[z]$ has variance $\OO(tk)=\OO((\eps k)^2)$,
the estimate has additive error $\OO(\eps k)$ with probability at least
$0.9$, say, by Chebyshev's inequality.
We can repeat $\Theta(\log n)$ times and take the median, which
achieves additive error $\OO(\eps k)$ w.h.p.\ by the
Chernoff bound.
\end{proof}

\subsection{Applying to Text-to-Pattern Hamming Distances}

Finally, we connect the Text-to-Pattern Hamming Distances to colored counting $3$SUM\@.
The connection is not difficult to see:
For each $a\in [n]$, put $a$ in the set $A$ with color $T[a]$, and
for each $b\in [m]$, put $b$ in the set $B$ with color $P[b]$, where $T$ and $P$ are the given text and pattern strings.
The number of matches between $P$ and $T[i\dd i+m-1]$ is precisely
$|\{(a,b)\in A\times B:\ a-b=i,\ \COL(a)=\COL(b)\}|$.  So, we can apply 
Lemma~\ref{lem:approx:col3sum} (with $U=n$, after negating $B$) to approximate the number of matches,
and thus also the number of mismatches, with variance/additive error dependent on $n$.

However, we want additive error $O(\eps k)$ for distances bounded by $k$.
To this end, we apply a technique from known exact algorithms:
Gawrychowski and Uzna\'{n}ski \cite{GawrychowskiU18} reduced $k$-bounded Text-to-Pattern Hamming Distances to $O(n/m)$ instance of the same problem for text and pattern strings of
length $O(m)$ that have \emph{run-length encodings} bounded by $O(k)$---in other words,
the text and pattern strings are concatenations of $O(k)$ blocks of identical characters (\emph{runs}).  
The reduction takes near-linear time, and preserves
additive approximation.  When mapped to the above colored sets $A$ and $B$, these blocks become 
monochromatic intervals.  So, we can apply Lemma~\ref{lem:RLE} to approximate the number of matches with additive error $O(\eps k)$,
and thus 
the number of mismatches with additive error
$O(\eps k)$, and
immediately obtain:

\begin{theorem}
\label{thm:approx_cap_k}
The approximate $k$-bounded Text-to-Pattern Hamming Distances problem %
additive error $O(\eps k)$
can be solved by a Monte Carlo algorithm in time
\[ \OO\left(\frac{n}{m} \cdot \left( \eps^{-1+\beta} m^{1-\beta}k^{\beta} +
\eps^{-1-\beta} m^{1+\beta} / k^{2\beta} + m
\right)\right), \qquad
\mbox{where $\beta := (3-\omega)/(5+\omega)$.}
\]
\end{theorem}
\cref{thm:approx_cap_k} implies 
\cref{thm:approxhammingmain}, which we recall below:
\approxhammingmain*
\begin{proof}

For all shifts with Hamming distance $k \le \eps^{-\gamma}\sqrt{m}$, we can use a known
exact algorithm running in $\OO(\frac{n}{m}\cdot (m+k\sqrt{m}))\le \OO(\eps^{-\gamma} n)$ time \cite{GawrychowskiU18,ChanGKKP20}. 
Otherwise, the bound in \cref{thm:approx_cap_k}
is at most $\OO(\frac{n}{m}\cdot (\eps^{-1+\beta}m + \eps^{-1-\beta+2\beta\gamma}m))$.
We thus obtain an upper bound
of $\OO(\eps^{-\gamma} n)$ in all cases, by setting $\gamma:=(1+\beta)/(1+2\beta) = 8/(11-\omega) < 0.928$.
We run the entire algorithm for every $k$ that is a power of~2.
\end{proof}

\begin{remark}
With more tedious calculations, the exponent 0.928 can likely be improved by using
known bounds on rectangular matrix multiplication, but the improvement would be tiny.  If $\omega=2$, the above bound is $\OO(\eps^{-8/9}n)$.  
Note that in the critical case when $k=\Theta(m)$ and $\eps^{-1}=\sqrt{m}$ (which we have mentioned earlier),
the bound in Theorem~\ref{thm:approx_cap_k} is actually a little better: $\OO(nm^{(1-\beta)/2})$, which is $\OO(nm^{3/7})$ if $\omega=2$.
\end{remark}

\section{Randomized Exact Text-to-Pattern Hamming Distances}

\newcommand{\COUNTLIVE}{\mbox{\sc count}_{\mbox{\scriptsize live}}}

In this and the next section, we turn to exact algorithms for Text-to-Pattern Hamming Distances.

\subsection{\texorpdfstring{$X+Y$}{X+Y} Lemma}
The key ingredient of our new algorithm is the following lemma on computing the sumset $X+Y$, along with the multiplicities of its elements,
for sets $X$ and $Y$ of $n$ elements in a bounded integer universe
(this is equivalent to computing convolutions of sparse binary vectors).  
\begin{lemma}
\label{lem:x+y}
Given two (multi)sets $X$ and $Y$ of $n$ elements in $[U/2]$ with $2U<n^2$,
we can compute $\COUNT[z]:=|\{(x,y)\in X\times Y: x+y=z\}|$ for every $z\in [U]$ by a Las Vegas algorithm
in $O(U\log(n^2/U))$ expected time.
\end{lemma}
Note that \cref{lem:x+y} improves over the standard $O(U\log U)$ bound by FFT, when $n^2$ is not too large compared to $U$.
Observe that the average count over all $z\in [U]$ is at most $n^2/U$, which is small in the regime of interest here.  It is tempting to simply apply a bit-packed version of FFT (Lemma~\ref{lem:packed_FFT}), but the challenge here is that there can be a mixture of elements
with small and (possibly very) large counts in the sumset, and we don't know which elements have small or large counts in advance.

Our algorithm needs the almost-linear hash family by Dietzfelbinger \cite{Dietzfelbinger96}, which was also used by Baran, Demaine, and  P\v{a}tra\c{s}cu \cite{BaranDP08} in their $3$SUM algorithm. The following statement was proved in~\cite{BaranDP08}.
\begin{lemma}[Almost-linear hash family \cite{BaranDP08}]
\label{lem:hash}
Let $L\le U$ be powers of two. There is a family of hash functions $f\colon [U]\to [L]$ with the following properties:

\begin{enumerate}[(i)]
\item For all $x,y\in [U]$, $f(x)+f(y)-f(x+y)\in\D_f$ for some fixed set $\D_f$ of $O(1)$ size. \label{item:falmostlinear}
\item For any fixed $x,y,z$ with $x+y\neq z$, $\Pr_f[f(x)+f(y)-f(z)\in\D_f]\le O(1/L)$. \label{item:ffewfalsepositive}
\item Sampling and evaluating $f$ take $O(1)$ time.
\end{enumerate}
\end{lemma}

Although the above hash family has been used successfully for various problems
related to convolutions, new technical issues arise when applying them to \emph{counting} problems: even though we know that $h(x)+h(y)-h(x+y)$ lies in a
set of $O(1)$ size for a hash function $h$, the precise value is still unpredictable. 
(If one uses another popular almost linear hash family, $h(x)=x\bmod p$ for random primes $p$, then this issue goes away, but collision probabilities increase by a $\log U$ factor, which we cannot afford to lose here.) 
 Below, we prove new additional properties about (one version of) Dietzfelbinger's hash family, stating that for most ``good'' pairs $(x,y)$, the value of $h(x)+h(y)-h(x+y)$ can be determined precisely by looking at some short labels
$\tau_h(x)$ and $\tau_h(y)$ of $x$ and $y$.

\begin{lemma}[Almost-linear hash family with somewhat predictable errors]
\label{lem:modifiedhash}
Let $q\le V\le U$ be powers of two. There is a family of hash functions $h\colon [U]\to [V]$ with the following properties:

\begin{enumerate}[(i)]
\item For any fixed $z,z'$ with $z\neq z'$, $\Pr_h[h(z)=h(z')]\le O(1/V)$.
\label{item:collisionprob}
\item There is a fixed set $\Delta_h$ of $O(1)$ size, and mappings $\tau_h\colon [U]\rightarrow [q^{O(1)}]$ and
$\phi_h\colon [q^{O(1)}]\times [q^{O(1)}]\rightarrow \D_h\cup\{\mbox{undefined}\}$ such that for all $x,y$, 
$h(x)+h(y)-h(x+y)=\phi_h(\tau_h(x),\tau_h(y))$ if $\phi_h(\tau_h(x),\tau_h(y))$ is defined.  
Call $(x,y)$ \emph{good} if $\phi_h(\tau_h(x),\tau_h(y))$ is defined, and \emph{bad} otherwise.
\label{item:mappingprediction}
\item For any fixed $x,y \in [U/2]$, $\Pr_h[\mbox{$(x,y)$ is bad}]\le  O(1/q)$.
\label{item:badprob}
\item Sampling and evaluating $h$ take $O(1)$ time. The mappings $\tau_h,\phi_h$ can also be computed in $O(1)$ time.
\end{enumerate}
\end{lemma}

\begin{proof}
For a nonnegative integer $a$, let $\bin(a)_i$ denote the $i$-th binary bit of $a$ (for example, $\bin(a)_0 = a \bmod 2$). For $i\ge j\ge 0$, let $\bin(a)_{i:j}$ denote the concatenation $\bin(a)_i\circ \bin(a)_{i-1}\circ \cdots \circ \bin(a)_j$.
Let $\low(a)$ denote the smallest $i$ such that $\bin(a)_i=1$; define $\low(0)=+\infty$.

Let $U=2^w$, $V=2^\ell$, and $q = 2^k$. %
We define the hash family  as follows, similar to \cite{Dietzfelbinger96,BaranDP08}:     Pick a random odd $r\in [2^{w+\ell}]$.  For $x\in [2^w]$, define $h(x)$ as the following $\ell$-bit integer,
\[h(x) = \bin(r\cdot x)_{w+\ell-1:w}.\]

Observe that, when $\low(x) = j$, we have $\bin(r\cdot x)_{j}=1$, $\bin(r\cdot x)_{j'}=0$ for all $j'<j$, and $\bin(r\cdot x)_{w+\ell-1:j+1}$ is a uniform random bit string.

We first prove \cref{item:collisionprob}. Given $0\le z'<z< 2^w$, note that $h(z)=h(z')$ implies $(r\cdot z - r\cdot z')\equiv v \pmod{2^{w+\ell}}$ for some integer $v$ with $|v|< 2^w$, which then implies $\bin(r\cdot (z-z'))_{w+\ell-1:w}$ is either the all-$0$ or all-$1$ $\ell$-bit string. 
Since $\low(z-z')\le w-1$, we know
 $\bin(r\cdot (z-z'))_{w+\ell-1:w}$ is a uniform random $\ell$-bit string. Thus, $h(z)=h(z')$ happens with probability at most $2/2^{\ell} = O(1/V)$.

Now we prove \cref{item:mappingprediction}. 
First note that the error $  h(x+y)-h(x)-h(y)\in \{0,1, -2^\ell, -2^\ell+1\}$, where the $+1$ term appears if and only if adding $r\cdot x$ and $r\cdot y$ generates a carry to the $w$-th bit, and the $-2^\ell$ term appears if and only if this addition generates a carry to the $(w+\ell)$-th bit.
In order to predict these two carry bits, we can define $\tau_h$ as the following pair of $k$-bit strings,
\[\tau_h(x) =(\bin(r\cdot x)_{w-1:w-k},\; \bin(r\cdot x)_{w+\ell-1:w+\ell-k}). \]
Observe that adding $r\cdot x$ and $r\cdot y$ generates a carry to the $w$-th bit if and only if \[\bin(r\cdot x)_{w-1:0} + \bin(r\cdot y)_{w-1:0} \ge 2^w.\]
By looking at the sums of their prefixes $\bin(r\cdot x)_{w-1:w-k}$ and $\bin(r\cdot y)_{w-1:w-k}$ (which are the first component in $\tau_h(x)$ and $\tau_h(y)$), we can unambiguously tell whether this inequality holds, except in the case where this sum $\bin(r\cdot x)_{w-1:w-k} + \bin(r\cdot y)_{w-1:w-k} $ happens to equal the all-$1$ $k$-bit string. In this case we let $\phi_h(\tau_h(x),\tau_h(y))$ be undefined.  Similarly, we use the second component of $\tau_h(x)$ and $\tau_h(y)$ to predict the carry to the $w+\ell$-th bit, and let $\phi_h(\tau_h(x),\tau_h(y))$ be undefined if this fails. This proves \cref{item:mappingprediction}.

Now we bound the probability that $\bin(r\cdot x)_{w-1:w-k} + \bin(r\cdot y)_{w-1:w-k} $ equals the all-$1$ $k$-bit string. Note that this event implies $\bin(r\cdot (x+y))_{w-1:w-k}$ is either the all-$1$ or all-$0$ string. We analyze each of these two cases, as follows:
\begin{itemize}
    \item The all-$1$ case may happen only if $\low(r\cdot (x+y)) = \low(x+y) \le w-k$. In this case, $\bin(r\cdot (x+y))_{w-1:w-k+1}$ is a uniformly random $(k-1)$-bit string, so the probability that it equals the all-$1$ string is at most $1/2^{k-1} = O(1/q)$.
    \item  Since we know $\low(r\cdot (x+y)) =\low(x+y) \le w-1$ from the assumption that $x,y\in [U/2]$, we know that the all-$0$ case may happen only if $\low(r\cdot (x+y)) = \low(x+y) \le w-k-1$. In this case, $\bin(r\cdot (x+y))_{w-1:w-k}$ is a uniformly random $k$-bit string, so the probability that it equals the all-$0$ string is at most $1/2^{k} = O(1/q)$.
\end{itemize}  
Hence, we know the probability that $\bin(r\cdot x)_{w-1:w-k} + \bin(r\cdot y)_{w-1:w-k} $ equals the all-$1$ $k$-bit string is at most $O(1/q)$.

The probability that $\bin(r\cdot x)_{w+\ell-1:w+\ell-k} + \bin(r\cdot y)_{w+\ell-1:w+\ell-k} $ equals the all-$1$ string can be similarly bounded by  $O(1/q)$.  Then \cref{item:badprob} is proved by a union bound.
\end{proof}

Now we prove \cref{lem:x+y} by a new algorithm that uses a clever combination of hashing, bit-packed FFT, and recursion.  We use hashing to reduce the number of live elements (elements whose counts are not yet known), but interestingly, we also use hashing and an extra recursive call to identify candidates for light elements (elements whose counts are small).  Fortunately, since the number of live elements decreases at a super-exponential rate, the extra recursive calls do not blow up the final time bound, as we will see.

\begin{proof}
Our algorithm uses recursion. It recursively solves the following ``partial version'' of the original problem: There are $M$ \emph{live} elements $z\in [U]$ for which we are required to compute $\COUNT[z]$. For the remaining $U-M$ non-live elements $z\in [U]$, the correct values of $\COUNT[z]$ are already known and given to us.
Let $T(n,M,U)$ be the expected time complexity of the problem of computing $\COUNT[z]$ for all the $M$ live elements $z\in [U]$.
The original problem corresponds to the case where all elements are live and $M=U$; so we want to bound $T(n,U,U)$.

\paragraph{Step 1: classify elements as light or heavy.}
Let $L,s $ be parameters. 
Choose a random function $f\colon [U]\rightarrow [L]$
from a family of almost-linear hash functions from \cref{lem:hash}.
Let $\COUNT_f[z]:=|\{(x,y)\in X\times Y: f(x)+f(y)-f(z)\in\D_f\}|$.
Call $z$ \emph{light} if $\COUNT_f[z] < 2sn^2/L$, and \emph{heavy} otherwise.
By \cref{item:falmostlinear} of \cref{lem:hash}, if $z$ is light, then $\COUNT[z]\le \COUNT_f[z]< 2sn^2/L$.

To determine which elements are light or heavy, we recursively
solve the problem for the (multi)sets $f(X)$ and $f(Y)$, in $T(n,O(L),O(L))$ time, and obtain $c_f[k]:=|\{(x,y)\in X\times Y: f(x)+f(y)=k\}|$ for all $k$. Then 
we can obtain $\COUNT_f[z]=\sum_{\delta\in\D_f}c_f[f(z)+\delta]$.

We now bound the number of heavy live elements in two cases:
\begin{enumerate}[(i)]
    \item 
First, since $\sum_{z} \COUNT[z] = |X||Y|=n^2$, the number of elements $z$ with $\COUNT[z]\ge sn^2/L$ is at most $L/s$. \label{item:case1}
\item \label{item:case2}
Now we fix a live element $z$ with $\COUNT[z]< sn^2/L$.  Note that $\COUNT_f[z] - \COUNT[z]$ is 
the number of $(x,y)\in X\times Y$ with $x+y\neq z$ and $f(x)+f(y)-f(z)\in \D_f$, which has expectation
at most $O(n^2/L)$ by \cref{item:ffewfalsepositive} of \cref{lem:hash}.
By Markov's inequality, the probability that this number exceeds $sn^2/L$ is $O(1/s)$.
Thus, the probability that $z$ is heavy is $O(1/s)$.  
\end{enumerate}
Summing up Case~\ref{item:case1} and Case~\ref{item:case2} (over all $M$ live elements), the expected number of heavy live elements is $O(L/s + M/s)$.
By Markov's inequality, we can guarantee that this bound holds after $O(1)$
expected number of trials.

\paragraph{Step 2: classify elements as isolated or non-isolated.}
Let $t$ and $q$ be parameters.  
Independently choose another random function
$h\colon [U] \rightarrow [tM]$ from a family of modified almost-linear hash functions from \cref{lem:modifiedhash} (with parameter $q$).
Let $\COUNTLIVE[z]:=|\{z'\ \mbox{live}:\ z'\neq z,\ h(z')=h(z)\}|$.
Call a live element $z$ \emph{isolated} if $\COUNTLIVE[z]=0$, and \emph{non-isolated} otherwise.

For a fixed live element $z$, the expectation of $\COUNTLIVE[z]$ is 
$O(M\cdot 1/(tM))=O(1/t)$ by \cref{item:collisionprob} of \cref{lem:modifiedhash}.  By Markov's inequality, the probability that $z$ is non-isolated is
$O(1/t)$.  So, the expected number of non-isolated live elements is $O(M/t)$.  By Markov's inequality, we can guarantee that this bound holds after $O(1)$
expected number of trials.

\paragraph{Step 3: compute counts for isolated light elements.}
\newcommand{\cbad}{c_{\mbox{\scriptsize\rm bad}}}
\newcommand{\cgood}{c_{\mbox{\scriptsize\rm good}}}
Define
\begin{eqnarray*}
 c[k] &:=& \sum_{z:\ h(z)=k}\COUNT[z]\\ 
 &=& |\{(x,y)\in X\times Y: h(x+y)=k\}|.
\end{eqnarray*}
To compute $c[k]$, we decompose it into $c[k] = \cbad[k] + \cgood[k]$ and compute separately,
where $\cbad[k] := |\{(x,y)\in X\times Y\ \mbox{bad}: h(x+y)=k\}|$ 
and $\cgood[k] := |\{(x,y)\in X\times Y\ \mbox{good}: h(x+y)=k\}|$ (see the definition of good and bad in \cref{lem:modifiedhash}).
Let $X^{(\alpha)}:=\{x\in X: \tau_h(x)=\alpha\}$ and 
$Y^{(\beta)}:=\{y\in Y: \tau_h(y)=\beta\}$. We preprocess sets $X^{(\alpha)}$ for all $\alpha \in [q^{O(1)}]$ (and sets $Y^{(\beta)}$ for all $\beta\in [q^{O(1)}]$) in $O(n+q^{O(1)})$ time. Then,
\begin{itemize}
    \item 
We can compute $\cbad[k]$ for all $k\in [tM]$ by
examining each $\alpha,\beta\in [q^{O(1)}]$ such that
$\phi_h(\alpha,\beta)$ is undefined, and enumerating all
$(x,y)\in X^{(\alpha)}\times Y^{(\beta)}$, and incrementing
the counter for $h(x+y)$.  Since each pair $(x,y)\in X\times Y$ is bad with $O(1/q)$ probability by \cref{item:badprob} of \cref{lem:modifiedhash}, the expected total time
is $O(q^{O(1)}+ n^2/q)$.
\item 
Fix a prime $p\in [2sn^2/L, 4sn^2/L]$.\footnote{The primes used by this recursive algorithm can be generated and fixed at the very beginning, in $\polylog(n)$ Las Vegas time. }
We can compute $\cgood[k]\bmod p$ for all $k\in [tM]$, by
examining each $\alpha,\beta\in [q^{O(1)}]$ such that
$\phi_h(\alpha,\beta)$ is defined, and
computing $\cgood^{(\alpha,\beta)}[k] := |\{(x,y)\in X^{(\alpha)}\times Y^{(\beta)}:\ h(x)+h(y)=k+\phi_h(\alpha,\beta)\}|\bmod p$,
which reduces to a convolution problem for the multisets 
$h(X^{(\alpha)})$ and
$h(Y^{(\beta)})$ in $[tM]$, done modulo $p$.
By \cref{lem:packed_FFT}, each convolution takes $O(tM\log (sn^2/L))$ time.
The total time over all $\alpha,\beta$ is $O(q^{O(1)}tM\log (sn^2/L))$.
\end{itemize}

Now, we know $c[k]\bmod p$ for all $k\in [tM]$ by summing up $\cbad[k]$ and $\cgood[k]\bmod p$.

For each isolated live element $z$,
we can compute $\COUNT[z]\bmod p$ by taking $c[h(z)]$ and 
subtracting $\COUNT[z']$ for all $z'$ with $z'\neq z$ and
$h(z')=h(z)$.
  Since $z$ is isolated, all such elements $z'$ are not live and
thus their $\COUNT[z']$ values are known.   The expected number of such elements $z'$
is $O(U/(tM))$, by \cref{item:collisionprob} of \cref{lem:modifiedhash}.
  Hence, the total expected time over all isolated live elements $z$
is $O(M\cdot U/(tM))=O(U/t)$.

If $z$ is light, $\COUNT[z]$ is the same as $\COUNT[z]\bmod p$.
Thus, we have computed $\COUNT[z]$ for all isolated light live elements $z$.

\paragraph{Step 4: compute counts for non-isolated elements and
heavy elements.}
The remaining live elements are non-isolated or heavy.  
We have already shown that there are $O(L/s + M/s + M/t)$ such elements.
We recursively solve the problem for these elements
in $T(n,O(L/s+M/s+M/t),U)$ time.

\paragraph{Analysis.}  We obtain the following recurrence: 
\begin{eqnarray*}
 T(n,M,U) &\le& O(1)T(n,O(L),O(L)) + T(n,O(L/s+M/s+M/t),U) \\
 && {} + O\big(q^{O(1)}tM\log (sn^2/L) + n^2/q + U/t\big).
\end{eqnarray*}
Let $M=U/r$. Set $L=U/r^{3/2}$, $s=t=\sqrt{r}$, and $q=r^\eps$ for a sufficiently small constant $\eps>0$. Recall $n^2>2U$.  Then the recurrence simplifies to
\[  T(n,U/r,U)\ \le\ O(1) T(n,O(U/r^{3/2}),U) + O\big ((U/r^{1/2-O(\eps)})\log(n^2/U) + n^2/r^\eps\big ),
\]
where we absorbed the $r$ dependence in $\log(sn^2/L) = \log(r^2 n^2/U)$ into the $r^{O(\eps)}$ factor.
Expanding the recurrence gives
\begin{eqnarray}
T(n,U/r,U) &\le& O\left(\sum_{i=0}^\infty  
O(1)^i \left((U/r^{(3/2)^i (1/2-O(\eps))})\log(n^2/U)+ n^2/r^{(3/2)^i\eps}\right) \right)\nonumber\\
&=& O\big((U/r^{1/2-O(\eps)})\log(n^2/U)+ n^2/r^\eps\big ).  \label{eqn1}
\end{eqnarray}

Now we explain an alternative, simpler algorithm, to be used in the first level of recursion: 
Instead of running Steps 2--3, we just directly
compute the counts for the light live elements by packed FFT (\cref{lem:packed_FFT}) over the universe $U$ modulo a $p\in [2sn^2/L, 4sn^2/L]$,
and recurse on the remaining heavy elements. The recurrence is
\begin{eqnarray*}
 T(n,M,U) &\le& O(1)T(n,O(L),O(L)) + T(n,O(L/s+M/s),U) + O(U\log (sn^2/L)).
\end{eqnarray*}
Set $M=U$, $L=U/r_1$, and $s=r_1$.  Then
\begin{eqnarray*}
 T(n,U,U) &\le& O(1) T(n,O(U/r_1),U) + O(U\log (r_1^2n^2/U))\\
 &\le & O((U/r_1^{1/2-O(\eps)})\log(n^2/U)+ n^2/r_1^\eps + U\log(r_1^2n^2/U))\qquad\mbox{by (\ref{eqn1}).}
\end{eqnarray*}
Finally, setting $r_1=(n^2/U)^{1/\eps}$ yields
$T(n,U,U)=O(U\log(n^2/U))$.
\end{proof}

\subsection{Text-to-Pattern Hamming Distances}
Our exact algorithm for the Text-to-Pattern Hamming Distances problem (\cref{thm:exacthammingmain}) now easily follows from \cref{lem:x+y}.
\exacthammingmain*
\begin{proof}
For each character $c\in\Sigma$,
let $A_c=\{a\in [n]:T[a]=c\}$ and $B_c=\{b\in[m]:P[b]=c\}$,
and $n_c=|A_c|+|B_c|$.
Note that
$\sum_c n_c = O(n)$.
The number of matches between $P$ and $T[i\dd i+m-1]$ is precisely $\sum_c |\{(a,b)\in A_c\times B_c: a-b=i\}|$.
So, the problem reduces to solving an instance of the problem from \cref{lem:x+y} (after negating $B_c$) with $n_c$ elements and universe size $n$, for each character $c$. 

If $n_c \le \sqrt{2n}$, we can solve the problem by brute-force in $O(n_c^2)$ time. It is straightforward to bound the total running time of this case by $O(n^{3/2})$. 

For $n_c > \sqrt{2n}$, we apply \cref{lem:x+y}. Consider any $\ell = 1, \ldots, \lceil \log(\sqrt{n/2}) \rceil$, and consider $n_c$ that is between $2^{\ell-1}\sqrt{2n}$ and $2^{\ell} \sqrt{2n}$. The number of such $c$ is $O(\frac{\sqrt{n}}{2^\ell})$. Moreover, each time we call \cref{lem:x+y}, if its running time exceeds twice its expectation, we rerun \cref{lem:x+y}. By a standard application of the Chernoff bound, the total number of reruns is $O(\max\{\frac{\sqrt{n}}{2^\ell}, \log n\})$ w.h.p. Therefore, w.h.p., the running time contributed by these $n_c$ is 
$$O\left(\max\left\{\frac{\sqrt{n}}{2^\ell}, \log n\right\} \cdot n \cdot \log((2^{\ell}\sqrt{2n})^2/n)\right) = O\left(\frac{\ell}{2^\ell} \cdot n^{3/2} + \ell n \log n\right).$$ 
Summing up over all $\ell = 1, \ldots, \lceil \log(\sqrt{n/2}) \rceil$ gives the $O(n^{3/2})$ running time. 

Finally, by breaking the problem into $O(n/m)$ instances of size $O(m)$,
the time bound becomes $O((n/m)\cdot m^{3/2})$.
\end{proof}

\subsection{\texorpdfstring{$k$}{k}-Mismatch}

Our technique also improves the previous algorithm for the $k$-mismatch problem. In fact, we only need to replace the use of FFT in \cite{ChanGKKP20}'s $O\big (n + \min \big (\frac{nk}{\sqrt{m}}\sqrt{\log m},\frac{nk^2}{m}\big )\big )$-time algorithm with our new \cref{lem:x+y}. 

\kmismatchmain*
\begin{proof}[Proof Sketch]
    The bottleneck of \cite{ChanGKKP20}'s algorithm lies in the following task: given $2t$ sparse sequences $f_1, \ldots, f_t, g_1, \ldots, g_t$ whose supports are all in $[n]$, and the total size of their supports is $O(k)$, compute (a sparse representation of) $f_i \star g_i$ for every $i$. Additionally, all nonzero entries of $f_i$ and $g_i$ are either $0$ or $1$. The running time for this task in \cite[Lemma 7.8]{ChanGKKP20} is $O(k \min(k, \sqrt{n \log n}))$. It suffices to improve it to provide an $O(k \min(k, \sqrt{n}))$. 

    Let $n_i$ denote the sum of the support size of $f_i$ and $g_i$. Note that $\sum_i n_i = O(k)$. We can either compute $f_i \star g_i$ using brute-force or using \cref{lem:x+y}. Therefore the running time can be written as 
    \[ O\left(\sum_{i:\ n_i \le \sqrt{2n}} n_i^2\ + \sum_{i:\ n_i > \sqrt{2n}}
    n\log (n_i^2/n)\right)\ =\ O(k \min(k, \sqrt{n})).
    \]

    \vspace{-\bigskipamount}
\end{proof}

\subsection{Text-to-Pattern Dominance Matching}
 
 In the Text-to-Pattern Dominance Matching problem, we want to compute $|\{i: P[i] \le T[i+k]\}|$ for all $k$. For convenience in the following we solve the variant $|\{i: P[i] < T[i+k]\}|$ (which is without loss of generality).

We prove the first part of \cref{thm:exactdominancemain}.
\begin{theorem}
\label{thm:exactdominancerandom}
The Text-to-Pattern Dominance Matching problem can be solved by
a Las Vegas algorithm which terminates in $O(n\sqrt{m})$ time with high probability.
\end{theorem}

\begin{proof}
Let us sort all the $n+m\le 2n$  characters (we treat the same character on different locations as different) of the text and the pattern together. For every dyadic interval $I$ on this sorted array (without loss of generality, we assume the length of this array is a power of $2$), let $L$ be the set of indices of the characters in the pattern in the left half of the dyadic interval, and let $R$ be the set of indices of the characters in the text in the right half of the dyadic interval. It suffices to count the contribution of $(i, j) \in L \times R$, i.e., the number of $(i, j) \in L \times R$ with $P[i] < T[j]$ and $i+k=j$ for every $k$ whose count is nonzero. Because we sorted the characters, we already have $P[i] \le T[j]$ for every $(i, j) \in L \times R$. If $\{P[i]: i \in L\}$ and $\{T[j]: j \in R\}$ share some character $c$ (there can be at most one), let $L_c := \{i \in L: P[i] = c\}$ and $R_c := \{j \in R: T[j] = c\}$. Now it suffices to count the contributions from $(L \setminus L_c) \times R$ and $L_c \times (R \setminus R_c)$, each of which can be handled with convolution. Let the dyadic interval $I$ be of length $2^{\ell+1}$. If $2^{\ell} \le \sqrt{2n}$, we use brute-force to compute the convolution; otherwise, we use \cref{lem:x+y}. Summing over all dyadic intervals give the following running time:
$$O\left(\sum_{0 \le \ell \le \log(\sqrt{2n})} \frac{n}{2^\ell} \cdot (2^\ell)^2\ + \sum_{\log(\sqrt{2n}) < \ell \le  \log n } \frac{n}{2^\ell} \cdot n \log((2^{\ell})^2 / n) \right)\ =\ O(n \sqrt{n}).$$
Breaking the problem into $O(n/m)$ instances of size $O(m)$ gives the $O(n\sqrt{m})$ running time. 
As in the proof of Theorem~\ref{thm:exacthammingdetmain}, this can be made to hold w.h.p.\ by applying the Chernoff bound.
\end{proof}

\section{Deterministic Exact Text-to-Pattern Hamming Distances}

For our deterministic exact algorithm, we switch to a simpler approach to hashing, namely, taking a number mod $m_i$ for some choice of $m_i$ (instead of using
Dietzfelbinger's hash family).  We use the following known lemma by Chan
and Lewenstein~\cite{ChanL15}:

\begin{lemma}[\cite{ChanL15}]
\label{lem:ChanL15}
    Given a set $T \subseteq [U]$ of size $n$, there exists an $n\cdot 2^{O(\sqrt{\log n \log \log U})} \cdot \polylog(t, U)$ time deterministic algorithm that constructs $r=2^{O(\sqrt{\log n \log \log U})}$ integers $m_1, \ldots, m_r = n \cdot 2^{\Theta(\sqrt{\log n \log \log U})} \cdot \polylog(t, U)$, where  for every $x \in T$, there exists $i \in [r]$ such that no other $y \in T$ has $y \equiv x \pmod{m_i}$.
\end{lemma}

Note that we slightly adapted the results in \cite{ChanL15} in that the integers $m_1, \ldots, m_r$ now also have a lower bound. This is without loss of generality because if some integer is too small, we can multiply it with an appropriate factor.  

One particular application of \cref{lem:ChanL15} in \cite{ChanL15} is a $t \cdot 2^{O(\sqrt{\log n \log \log U})} \cdot \polylog(t, U)$ time deterministic algorithm for the sparse nonnegative convolution problem, in which we are given two sparse nonnegative sequences $A, B$, and we need to compute (a sparse representation of) their convolution $A \star B$, with the additional assumption that a small size-$t$ superset of the support of the output sequence is given. Bringmann and Nakos \cite{icalpBringmannN21} removed this assumption via recursion. We first closely follow the approaches in~\cite{ChanL15}  and \cite{icalpBringmannN21} to solve a problem similar to sparse nonnegative convolution.

\begin{lemma}
\label{lem:sparse_conv_minus_C}
    Given three integer sequences $A, B, C$ of length $U$ and a set $T \subseteq [U]$, 
    with the promise that $(A \star B - C)[i] \ge 0$ for every $i$ and $T$ is a superset of the support of $A \star B - C$, we can compute $A \star B$ in $$O(U) + t \cdot 2^{O(\sqrt{\log t \log \log U})}\cdot \polylog(t, U)$$ deterministic time, where $t = \max\{||A||_0, ||B||_0, ||T||_0\}$.  
\end{lemma}
\begin{proof}

We apply \cref{lem:ChanL15} on $T$ and $U$ and find $r=2^{O(\sqrt{\log t \log \log U})}$ integers $m_1, \ldots, m_r = t \cdot 2^{\Theta(\sqrt{\log t \log \log U})} \cdot \polylog(t, U)$. For each $k \in [r]$, we prepare two arrays $A_k$ and $B_k$, which are defined as $A_k[i] := \sum_{j \equiv i \bmod{m_k}} A[j]$ and $B_k[i] := \sum_{j \equiv i \bmod{m_k}} B[j]$. Then we compute the following array $D_k$ via FFT in $\OO(m_k)$ time:
$$D_k[i]\ := \sum_{j \equiv i \bmod{m_k}} (A \star B)[j] = (A_k \star B_k)[i] + (A_k \star B_k)[i+m_k].$$
Furthermore, for each $x \in T$, we find the integer $m_k$ such that for any other $y \in T$, $y \not \equiv x \pmod{m_k}$. 
The above takes $t \cdot 2^{O(\sqrt{\log t \log \log U})}\cdot \polylog(t, U)$ time. 

Next, for each $x \in T$ and the corresponding $m_k$, we compute the following value $$E[x]\ :=\ D_k[x \bmod{m_k}] - \sum_{j \equiv x \pmod{m_k}} C[j], $$
which equals
$$\sum_{j \equiv x \pmod{m_k}} (A \star B)[j]\ - \sum_{j \equiv x \pmod{m_k}} C[j]\ = \sum_{j \equiv x \pmod{m_k}} (A \star B - C)[j]. $$
Since there is no other $y \in T$ such that $y \equiv x \bmod{m_k}$, and $T$ is a superset of the support of $A \star B - C$, 
$$\sum_{j \equiv x \pmod{m_k}} (A \star B - C)[j]\ =\ (A \star B - C)[x]. $$
For any $x \not \in T$, we can simply set $E[x]$ to be $0$.  The time for computing $E[x]$ for each $x \in T$ is $O(\frac{U}{m_k}) = O(\frac{U}{t \cdot 2^{\Theta(\sqrt{\log t \log \log U})} \cdot \polylog(t, U)}) = O(\frac{U}{t})$, so the overall running time for computing $E$ is $O(\frac{U}{t} \cdot |T|) = O(U)$. 

Overall, $E = A \star B - C$, and we can compute $A \star B$ by adding $E$ and $C$ in $O(U)$ time.

\end{proof}

\begin{lemma}
\label{lem:det_X_plus_Y}
Given two (multi)sets $X$ and $Y$ of $n$ elements in $[U/2]$ with $2U<n^2$,
we can compute $\COUNT[z]:=|\{(x,y)\in X\times Y: x+y=z\}|$ for every $z\in [U]$ by a deterministic algorithm
in $O(U\log(n^2/U) + U \sqrt{\log n \log \log U})$ time.
\end{lemma}
\begin{proof}
Without loss of generality, we assume $U$ is a power of $2$.

First, if $\frac{n^2}{U} = n^{\Omega(1)}$, then directly applying FFT already achieves the claimed running time. Now we assume $\frac{n^2}{U} = n^{o(1)}$.

    Let $p$ be a prime whose range is to be determined later. We use  \cref{lem:packed_FFT} to compute $\COUNT[z] \bmod{p}$ for every $z \in [U]$ in $O(U \log p)$ time. 

    Then we aim to apply \cref{lem:sparse_conv_minus_C} using $X$ for $A$, $Y$ for $B$ and $\COUNT[z]  \bmod{p}$ for $C$. However, we need a small set $T$ that is a superset of the support of $A \star B - C$. To do so, we recursively solve the following problem. 

    Let $U' = U / 2$, and define $X' := \{x \bmod{\frac{U'}{2}} : x \in X\}$ and $Y' := \{y \bmod{\frac{U'}{2}} : y \in Y\}$. Then we compute $\COUNT'[z]:=|\{(x,y)\in X'\times Y': x+y=z\}|$ recursively in $Q(n, U / 2)$ time, where $Q(n, U)$ denotes the running time to solve the instance as described in the lemma statement. We let $T$ be $\{z \in U': \COUNT'[z] \ge p / 2\} + \{0, U' / 2, U', 3 U' / 2\}$. 

    We verify that $T$ is a superset of the support of $A \star B - C$. Note that if $(A \star B - C)[i] \ne 0$ for some $i$, then $\COUNT[i] \ge p$, or equivalently, $|\{(x,y)\in X\times Y: x+y=i\}| \ge p$. This implies 
    \begin{align*}
    &\ |\{(x',y')\in X'\times Y': x' + y' \equiv i \pmod{U' / 2}\}| \\
    = & \ \left|\left\{(x,y)\in X\times Y: \left(x \bmod \frac{U'}{2}\right)+\left(y \bmod \frac{U'}{2}\right) \equiv i \pmod{U' / 2}\right\}\right|\\
    \ge & \ p.
    \end{align*}
    Therefore,  $\COUNT'[i \bmod{\frac{U'}{2}}] \ge p / 2$ or $\COUNT'[(i \bmod{\frac{U'}{2}}) + (U' / 2)] \ge p / 2$. In either case, $i$ will be included in $T$. 

    Additionally, it is easy to see that $|T| = O(n^2 / p)$, as $\sum_z \COUNT'[z] \le n^2$. Then we apply \cref{lem:sparse_conv_minus_C}, using $X$ for $A$, $Y$ for $B$, $\COUNT[z]  \bmod{p}$ for $C$ and $T$. The running time is $O(U) + t \cdot 2^{O(\sqrt{\log t \log \log U})} \cdot \polylog(t, U)$ for $t = \max\{n, \frac{n^2}{p}\}$. Also, as $2U < n^2$, we can simplify the running time to $O\left(U + t \cdot 2^{c \sqrt{\log t \log \log U})}\right)$ for some constant $c$.  

    The overall running time is thus 
    \[
    Q(n, U / 2) + O\left(U \log p + t \cdot 2^{c \sqrt{\log t \log \log U})} \right)
    \]
    for $t = \max\{n, \frac{n^2}{p}\}$. Picking $p$ from $\Theta\left(\frac{n^2}{U} \cdot 2^{c \sqrt{\log n \log \log U}} \right)$ (it only takes $n^{o(1)}$ time to find such a prime as $\frac{n^2}{U} = n^{o(1)}$) simplifies the running time to 
    \[
    Q(n, U / 2) + O\left(U\log(n^2/U) + U \sqrt{\log n \log \log U}\right),
    \]
    and it is easy to see that the recursion only incurs a constant factor blowup to the running time. 
\end{proof}

\exacthammingdetmain*
\begin{proof}
Let $n_c$ be the number of occurrences of character $c$ in the text and pattern.  Note that
$\sum_c n_c = O(n)$.
As we know, the problem reduces to solving an instance of the problem from \cref{lem:det_X_plus_Y} with $n_c$ elements and universe size $n$, for each character $c$. 

For character $c$ with $n_c \le n^{1/2}(\log n \log \log n)^{1/4}$, we use the brute-force $O(n_c^2)$ time algorithm. For character $c$ with $n_c > n^{1/2}(\log n \log \log n)^{1/4}$, we use the algorithm from \cref{lem:det_X_plus_Y} which runs in $O(n \log(n_c^2 / n) + n \sqrt{\log n_c \log \log n})$ time. The overall running time is $O(n^{3/2} (\log n \log \log n)^{1/4})$.

Finally, by breaking the problem into $O(n/m)$ instances of size $O(m)$,
the time bound becomes $O((n/m)\cdot m^{3/2} (\log m \log \log m)^{1/4})$.
\end{proof}

The deterministic algorithm for Text-to-Pattern Dominance Matching also easily follows from \cref{lem:det_X_plus_Y}. Its proof is identical to the proof of \cref{thm:exactdominancerandom}, except that we replace \cref{lem:x+y} with \cref{lem:det_X_plus_Y} and update the running time analysis properly.  

\begin{theorem}
The Text-to-Pattern Dominance Matching problem can be solved by
a deterministic algorithm in $O(n\sqrt{m} (\log m\log \log m)^{1/4})$ time.
\end{theorem}

\begin{remark}
It would be natural to build \cref{lem:sparse_conv_minus_C} on Bringmann, Fischer and Nakos's more efficient algorithm \cite{BringmannFN22} for sparse nonnegative convolution that runs in $\OO(t)$ time, in order to further improve our deterministic algorithm for exact Text-to-Pattern Hamming Distances. The difficulty in this approach is  that, to implement \cite{BringmannFN22}'s idea, we have to view $C$ as a degree $U$ polynomial and evaluate it on $t$ carefully chosen points. It is unclear how to do this evaluation in $o(U \log U)$ time ($O(U \log U)$ is the naive bound for \cref{lem:det_X_plus_Y} via FFT). 
\end{remark}

\section{Equivalence with a Variant of \texorpdfstring{$3$}{3}SUM}

In this section, we prove \cref{thm:equivalence}, which we recall below:

\equivalence*

\begin{proof}
The forward direction is implied by the proof of \cite[Theorem 2.17]{BringmannN20}. For completeness, we include this simple proof. Suppose we have an algorithm $\mathcal{A}$ for \cref{prob:$3$SUM_variant} in $T(N)$ time and we are given a Text-to-Pattern Hamming Distances instance with text $T$ and pattern $P$ where $n = O(m)$. For every $i \in [m]$, we add a number $-2nP[i]-i$ to a set $A$. For every $i \in [n]$, we add a number $2nT[i]+i$ to $B$. Finally, let $C = [n]$. Then we run algorithm $\mathcal{A}$ on sets $A, B, C$. It suffices to show that for every $i \in C$, the number of $(a, b) \in A \times B$ with $a+b = i$ is exactly the number of $j$ where $P[j] = T[i + j]$ (if this is the case, then the Hamming distance between $P$ and $T[i \mathinner{.\,.} i + m-1]$ is $m$ minus the count of $3$SUM solutions for $i \in C$). In order for some number $-2nP[j] - j \in A$ and some number $2nT[k] + k \in B$ to sum up to $i$, we must  have $P[j] = T[k]$ and $-j + k = i$. Therefore, the number of such pairs is exactly the number of $j$ where $P[j] = T[i + j]$. 

We next show the previously unknown backward direction. Suppose we have a $T(n)$-time algorithm $\mathcal{B}$ for Text-to-Pattern Hamming Distances with $n = O(m)$ and we are given an instance of \cref{prob:$3$SUM_variant}. First, we can negate all numbers in $A$, so that the task becomes finding the number of $(a, b) \in A \times B$ where $-a + b = c$ for every $c \in [N]$. Then we partition the sets $A, B$ in the following way: for any integer $g$, let $A_g := \{a \in A: gN \le a < (g+1)N\}$  and similarly let $B_g := \{b \in B: gN \le b < (g+1)N\}$ (we do not need to create $A_g$ or $B_g$ that is empty). In order for $-a+b \in [N]$, we only need to match numbers in $A_g$ with numbers in $B_{g-1}, B_g, B_{g+1}$. In the following, we only consider matching numbers in $A_g$ with numbers in $B_g$, and the other two cases can be handled similarly. 

 For every $g$ where $A_g'$ and $B_g'$ are nonempty, we sample a uniformly random shift $s_g \in [N]$. Let $A'_g := \{a - gN + s_g: a \in A_g\}$ and let $B'_g := \{b - gN + s_g: b \in B_g\}$ (this random shifts idea appeared in \cite{Lincoln0W20}). Now the problem becomes, for every $c \in [N]$, find the number of $g, a \in A'_g, b \in B'_g$ where $-a + b = c$. Note that all numbers in $A'_g$ and $B'_g$ are in $[2N]$. For each $i \in [2N]$, the expected number of times it appears in $A_g'$ and $B_g'$ over all $g$ is at most $\frac{1}{N} \sum_i (|A_g|+|B_g|) = O(1)$, so by the Chernoff bound, the number of times it appears in $A_g'$ and $B_g'$  is $O(\log N)$ w.h.p. For every $(x, y) \in \{1, \ldots, O(\log N)\}^2$, we create a Text-to-Pattern Hamming Distances instance as follows. Let the pattern $P_x$ be of length $2N$, initially consisting of unique characters at each position. Then for every $i$, if $A_g'$ is the $x$-th set (among all $A_g'$'s) that contains $i$, we set $P_x[i]$ to be $g$. Similarly, let the text $T_y$ be of length $3N$, initially consisting of unique characters at each position. Then for every $i$, if $B_g'$ is the $y$-th set (among all $B_g'$'s) that contains $i$, we set $T_y[i]$ to be $g$. Then we call algorithm $\mathcal{B}$ on each of these $O(\log^2 N)$ instances. For each shift $i \in [N]$, $2N$ minus the Hamming distance equals the number of $j$ where $P_x[j] = T_y[i+j]$, i.e., it is the number of $g, a \in A_g', b \in B_g'$ where $-a + b = i$, $A_g'$ is the $x$-th set (among all $A_g'$'s) that contains $a$, $B_g'$ is the $y$-th set (among all $B_g'$'s) that contains $b$. Summing over all $(x, y)$ gives exactly the quantity we seek for each $i \in [N]$. The overall running time of this algorithm is $O(T(N) \log^2 N)$. 
\end{proof}

\section{Open Problems}

We conclude with a few open questions:
\begin{itemize}
    \item For $(1+\eps)$-approximating Text-to-Pattern Hamming distances, what is the best possible dependence on $1/\eps$? Are there \emph{deterministic} algorithms faster than Karloff's $\OO(\eps^{-2}n)$ algorithm \cite{Karloff93}?
    \item Is there a $o(n\sqrt{m})$-time randomized algorithm for exact Text-to-Pattern Hamming Distances in the word-RAM model?  Is there an $O(n\sqrt{m})$-time deterministic algorithm?
    \item Do our algorithms generalize to Text-to-Pattern $\ell_p$ Distances? 
\end{itemize}

\bibliographystyle{alphaurl} 
\bibliography{main}

\appendix
\section{Slight Extension to Indyk's Reduction}
In this appendix, we observe a simple extension to Indyk's reduction from BMM to Text-to-Pattern Hamming Distances, so that instead of BMM we start from the following Equality Product problem: Given two $N\times N$ integer matrices $A$ and $B$, compute the $N\times N$ matrix $C$ where $C[i,j]=|\{k~|~A[i,k]=B[k,j]\}|$.

Equality Product has been studied in several papers \cite{MatIPL,LabibUW19,vnotes,ChanVX23,GoldS17,Lincoln0W20}. The fastest known algorithm \cite{YusterDom} for it runs in $O(N^{2.659})$ time using rectangular matrix multiplication \cite{GallU18, journals/corr/abs-2307-07970}. This  running time would be $O(N^{2.5})$ if $\omega=2$.

Equality Product is among the so called ``intermediate'' matrix products which seem to require $N^{2.5-o(1)}$ time (in the word-RAM model of computation with $O(\log N)$ bit words), even if $\omega=2$ (see \cite{Lincoln0W20,LabibUW19}).

Here we reduce Equality Product to Text-to-Pattern Hamming Distances, following Indyk's reduction closely.

Given $N\times N$ matrices $A$ and $B$, we create a text $T$ and a pattern $P$, both of length $\Theta(N^2)$ as follows.

First let us define our alphabet $\Sigma$.
For every $i,j\in [N]$, interpret $(j,A[i,j])$ as a new letter in $\Sigma$.
Similarly, for every $j,k\in [N]$, add letters $(j,B[j,k])$ to $\Sigma$.
So far, $\Sigma$ is a subset of $[n]\times \Z$.
Also let $\$$ be a new letter that does not appear in $\Sigma$ so far, adding it to $\Sigma$.

Encode each row $A_i$ of $A$ as a string $f(i) = (0,A[i,0])\odot (1,A[i,1])\odot\ldots\odot (N-1,A[i,N-1])$, where $\odot$ means concatenation.
Let the text be 
$$T=\$^{N^2}\odot f(0)\odot \$ \odot f(1)\odot \$ \odot \ldots \odot \$ \odot f(N-1)\odot \$^{N^2}.$$

Similarly, encode each column $B_k$ of $B$ as a string $g(k)=(0,B[0,k])\odot (1,B[1,k])\odot\ldots\odot (N-1,B[N-1,k])$.
Let the pattern be
$$P=g(0)\odot g(1)\odot \ldots g(N-1).$$

Note that the Hamming distance between $f(i)$ and $g(k)$ is exactly the number of $j$ for which $A[i,j]\neq B[j,k]$, so that $N-$the Hamming distance of $f(i)$ and $g(k)$ is exactly the number of $j$ for which $A[i,j]=B[j,k]$.

Similarly to Indyk's reduction, the $\$$ symbols in $T$ ensure that if we align $P$ with $T$ so that $f(i)$ is exactly aligned with $g(k)$, then there are no other symbols of $\Sigma$ that can be equal and aligned except those in $f(i)$ and $g(k)$, and so the Hamming distance between $T$ and $P$ for the corresponding shift equals $|P|-$the number of $j$ for which $A[i,j]=B[j,k]$.

The lengths $n$ and $m$ of $T$ and $P$ are both $\Theta(N^2)$. Thus, any algorithm that runs in $O(nm^{1/4-\epsilon})$ time for $\epsilon>0$ for Text-to-Pattern Hamming Distances would result in an $O(N^2\cdot N^{2/4-2\epsilon})=O(N^{2.5-2\epsilon})$ time algorithm for Equality Product.

The lower bound for Text-to-Pattern Hamming Distances would be higher if we assumed that the current best known algorithms for Equality Product are optimal. In particular, if Equality Product requires $N^{2.5+\delta-o(1)}$ time for some $\delta>0$, then the lower bound for Text-to-Pattern Hamming Distances becomes $nm^{1/4+\delta/2-o(1)}$.

\paragraph{Extension to Gawrychowski and Uzna\'{n}ski's reduction for $k$-mismatch.}
We remark that the same modification can be performed on the reduction by 
Gawrychowski and Uzna\'{n}ski \cite{GawrychowskiU18} for the $k$-mismatch problem, to give a conditional lower bound for $k$-mismatch of $((kn/\sqrt{m})\cdot (1/m^{1/4}))^{1-o(1)}$ which would hold even if $\omega=2$. Recall that the fastest algorithm runs in $O(n+kn/\sqrt m)$ time.

Gawrychowski and Uzna\'{n}ski presented a reduction from Boolean Matrix Multiplication of an $M'\times N$ matrix by an $N\times M$ matrix (for $M'\geq M\geq N$), to the $k$-mismatch problem for a text of length $n$ and a pattern of length $m$ where $m=M^2$, $n=M'M$ and $k=MN$.

Using our simple modification we immediately obtain a reduction from the Equality Product of an $M'\times N$ matrix by an $N\times M$ matrix to the $k$-mismatch problem for a text of length $n$ and a pattern of length $m$ where $m=M^2$, $n=M'M$ and $k=MN$.

The fastest algorithm for this rectangular Equality Product when $M\leq N^2$ runs in $O(M'N \sqrt M)$ time if $\omega=2$.

Suppose that $k$-mismatch has an algorithm running in time $O((kn/m^{3/4})^{1-\eps})$ time for some $\eps>0$. Then Equality Product of an $M'\times N$ matrix by an $N\times M$ matrix can be solved asymptotically in time
\[\left(\frac{(MN)\cdot(M'M)}{(M^2)^{3/4}}\right)^{1-\eps} = \left(M'N\sqrt{M}\right)^{1-\eps},\]
thus beating the best known running time for rectangular Equality Product in this setting, even if $\omega=2$.

\section{Polynomial Multiplication over \texorpdfstring{$\F_p$}{Fp} in word RAM}
\label{app:poly-mult}

\newcommand{\GF}[1]{\F_{#1}}

In this section, we prove \cref{lem:packed_FFT}, which we recall below: 
\packedFFT*

\paragraph{Indyk's original approach. } 
Indyk \cite{Indyk98a} claimed a proof of \cref{lem:packed_FFT} (originally described in the $p=2$ case) as follows: 
in the word RAM model with $\Theta(\log n)$-bit word length, we can pack $\ell = \Theta(\frac{\log n}{\log p})$ numbers in $\GF{p}$ to a single word, represented in $\GF{p^{2\ell-1}}$, which reduces the length of the arrays to $O(n / \ell)$. Then we essentially need to multiply two degree-$O(n / \ell)$ polynomials over $\GF{p^{2\ell-1}}$. 
Indyk \cite{Indyk98a} assumed that this multiplication could be done in $O((n/\ell) \log(n / \ell))$ time (where each field operation in $\GF{p^{2\ell-1}}$ takes constant time), which would imply the desired $O(n \log p)$ time bound. Today, it is known how to  perform this multiplication for $p=2$ in $O((n/\ell) \log(n / \ell))$ time \cite{LinAHC16}, but for general $p$, the current best algorithm for multiplying two degree-$m$ polynomials over a finite field uses $O(m \log m \cdot 2^{O(\log^* m)})$ field operations~\cite{HarveyHL17} (for finite fields with primitive roots of large smooth order, the textbook Cooley-Tukey FFT has a faster $O(m\log m)$ run time, but for general finite fields it is a major open question to remove this $2^{O(\log^* m)}$ factor; see e.g.,\ discussion in \cite{BenSassonCKL23}), so Indyk's original proof (combined with the state-of-the-art \cite{HarveyHL17} result directly as a black box) only implies an algorithm with slower $O(n \log p \cdot 2^{O(\log^* n)})$ time. 

\begin{remark}
We briefly discuss other papers that relied on Indyk's algorithm and are hence affected by this issue (but can be saved either by \cite{LinAHC16} or using our new proof of \cref{lem:packed_FFT}). 
\cite{CardozeS98} studied pattern matching for point sets and gave a randomized $O(n\log n)$ time algorithm. \cite{ChanH20a} studied the coin change problem and gave a randomized $O(t\log t)$ time decision algorithm. \cite{BringmannFN21} gave a  randomized $O(k\log k)$ time non-negative sparse convolution algorithm. These results \cite{CardozeS98,ChanH20a,BringmannFN21} only use the case $p=2$, so they can already be fixed by \cite{LinAHC16}.%
\end{remark}

In the rest of the section, we describe a more involved word RAM algorithm
that saves this additional $2^{O(\log^* n)}$ factor,  proving \cref{lem:packed_FFT}. It builds on the recursive algorithm of \cite{HarveyHL17}  and additionally uses bit tricks and table look-ups (in a similar spirit to the $O(n)$-time integer multiplication algorithm in the word RAM model by \cite{Furer14a}). 
This algorithm does not solve the aforementioned major question left open by \cite{HarveyHL17}, since it runs on the word RAM model.

First, multiplying two length-$n$ polynomials over $\F_p$ can be reduced to multiplying two $O(n\log(pn))$-bit integers
via a standard Kronecker substitution (see \cite[Section 2.6]{HarveyHL17}). The latter task can be done in $O(n\log(pn))$ time in word RAM \cite{Furer14a} using FFT with bit packing.
If $p\ge n^{0.01}$, then the running time is $O(n\log(pn)) = O(n\log p)$ as desired. Hence, in the rest of the section we assume $p\le n^{0.01}$.

At a high level, our algorithm uses the techniques from \cite{HarveyHL17}'s $\GF{p}$-polynomial multiplication algorithm with $O(n\log n \cdot 8^{\log^* n}\cdot \log p)$ bit complexity in the Turing Machine model.  Their algorithm has roughly $\log^* n$ levels of recursion, where each level exponentially decreases the length of the DFT\@. Here we adapt their algorithm to the word RAM model: we only need two levels of recursion to decrease the length of the DFT to sub-logarithmic, and then we look up the DFT results from preprocessed tables. We also need to use some bit tricks to speed up the DFT implementation.

\paragraph*{Number-theoretic lemmas.}
We call a positive integer \emph{$y$-smooth} if all of its prime divisors are less than or equal
to $y$.
We quote the following two theorems from \cite{HarveyHL17}.

\begin{lemma}[{\cite[Theorem 4.1]{HarveyHL17}}]
\label{lem:4.1}
There exist computable absolute constants $c_3 > c_2 > 0$ and $n_0 \in \N$ with
the following properties. Let $p$ be a prime, and let $n\ge n_0$. Then there exists an integer $\lambda$ in the interval
\[(\log n)^{c_2 \log \log \log n} <\lambda< (\log n)^{ c_3 \log \log \log n}, \]
and a $(\lambda + 1)$-smooth integer $M\ge n$, such that $M \mid p^{\lambda} - 1$. Moreover, given $p$ and $n$, we
may compute $\lambda$ and the prime factorization of $M$ in time $O((\log n)^{\log \log n} )$.
\end{lemma}

\begin{lemma}[{\cite[Theorem 4.6]{HarveyHL17}}]
\label{lem:4.6}
Let $p, n, \lambda, M$ be as in \cref{lem:4.1}. Let $R$ and $S$ be positive integers such that $\lambda < S < R < M$. Then there exist $(\lambda + 1)$-smooth integers $m_1,\dots, m_d$ with the following properties:
\begin{enumerate}
    \item $N := m_1\cdots m_d$ divides $M$ (and hence divides $p^{\lambda} - 1$).
    \item $R\le N\le (\lambda + 1)R$.
    \item $S\le m_i\le S^3$ for all $i$.
\end{enumerate}
Given $\lambda, S, R$, and the prime factorization of $M$, we may compute such $m_1,\dots, m_d$ (and
their factorizations) in time $\tilde O(\lambda^3)$.
\end{lemma}

We will use the following corollary which combines the two lemmas above.
\begin{cor}
\label{cor:combined}
Let $p$ be a prime, and let $n \ge n_0$ (where $n_0$ is an absolute constant). 
Then there exist integers $n' \in [n, 2n]$, $L \in (\log n)^{\Theta(\log \log \log n)}$, and $m_1,\dots,m_d$ (all computable in $n^{o(1)}$ time), such that:
\begin{itemize}
\item $N:= m_1\cdots m_d$ divides $p^{L} - 1$,
    \item $n' = NL$, and
    \item $\sqrt{L}/2 \le m_i \le L^3$.
\end{itemize}
\end{cor}
\begin{proof}
    Apply \cref{lem:4.1} with $p,n$ and obtain $\lambda, M$.  Then apply \cref{lem:4.6} with $S = \lambda+1$ and $R = \lfloor n/(\lambda+1)^2\rfloor $ to obtain $N:= m_1\cdots m_d$.

Let $n'$ be the smallest $n'\ge n$ such that $n'$ is a multiple of $N \lambda$. Since $N\lambda \le (\lambda+1)R\cdot \lambda\le (n/(\lambda+1))\cdot \lambda < n$,  we have $n'\le 2n$. Define $L:= n'/N$. Since $N\lambda \mid n'$,  we must have $\lambda \mid L$, and hence $N \mid p^{\lambda}-1 \mid p^L-1$. 

From $R\le N\le (\lambda+1) R$ (where $R= \lfloor n/(\lambda+1)^2\rfloor $) we have $n/N \in [\lambda+1,  2(\lambda+1)^2 ]$, and hence $L= n'/N \in [n/N, 2n/N] \subseteq [\lambda+1, 4(\lambda+1)^2] \subseteq (\log n)^{\Theta(\log \log \log n)}$.

From $S\le m_i \le S^3$, $S=\lambda+1$, and  $L\in [\lambda+1,4(\lambda+1)^2]$, we have $\sqrt{L}/2 \le m_i \le L^3$.

The running time for applying the two lemmas is $O((\log n)^{\log \log n}) + \tilde O(\lambda^3) \le O((\log n)^{\log \log n})\le n^{o(1)}$.
\end{proof}

Now we describe the parameters in the two levels of our algorithm for multiplying two degree-$n$ polynomials over $\GF{p}$. We will invoke \cref{cor:combined} twice with two different $n$'s (and always the same $p$).
\paragraph*{First level parameters.}
Let $n' =  N_1 L_1$ be returned by \cref{cor:combined} when plugging in $n$. Here $n'\in [n,2n]$, $L_1 \in (\log n)^{\Theta(\log \log \log n)}$, and $N_1 \mid p^{L_1}-1$.

We will consider the sub-problem of multiplying two degree-$L_1$ polynomials over $\F_p$.
At this point, we first address the easy case where $p\ge  L_1^{\Omega(1)}$, which can be solved without the second-level recursion: we simply do a Kronecker substitution to reduce this sub-problem to polynomial multiplication with integer coefficients.

\begin{lemma}[Second level (degenerate case)]
\label{lem:second-level-degenerate}
Suppose $p> L_1$.
Then multiplying two degree-$L_1$ polynomials over $\GF{p}$ can be done in $O((L_1\log L_1)\cdot \frac{\log p}{\log n})$ time.   
\end{lemma}
\begin{proof}
In this case, we may use standard Kronecker substitution to pack the coefficients into large integers (see \cite[Section 2.6]{HarveyHL17}). More specifically, we can pack $b$ coefficients in $\GF{p}$ into an integer of magnitude $O(p^2L_1)^b = p^{O(b)}$ as $p > L_1$. To fit each integer in a word, we can set $p^{O(b)} = n^{O(1)}$, so $b$ can be as large as $O(\frac{\log n}{\log p})$. 
Then the problem reduces to multiplying two degree-$\frac{L_1}{b}$ polynomials with integer coefficients (each fitting into one word), which can be done using FFT in $O(\frac{L_1}{b} \log (\frac{L_1}{b})) = O((L_1 \log L_1) \cdot \frac{\log p}{\log n})$ time as desired. 
\end{proof}

In the following, we need to prove \cref{lem:second-level-degenerate}  in the hard case where $p\le L_1$, via a second level of recursion.

\paragraph{Second level parameters.}
Assume $p\le L_1^{O(1)}$.
Let $L_1' =  N_2 L_2$ be returned by \cref{cor:combined} when plugging in $L_1$ in place of $n$. Here $L_1'\in [L_1,2L_1]$ and $L_2 \in (\log L_1 )^{\Theta(\log \log \log L_1)} \subseteq (\log \log n )^{\Theta(\log^{(4)} n)}$.

In the following we will work over the finite field $\GF{p^{L_2}}$.
Note that we can find a representation of $\GF{p^{L_2}}$ by finding an irreducible monic
polynomial of degree $L_2$, which can be done in  $\tO(L_2^4 p^{1/2}) \le \tO(L_2^4 L_1^{1/2}) = n^{o(1)}$ time  deterministically \cite{Shoup88}.
Since $N_2\mid  p^{L_2}-1$ by \cref{cor:combined}, we can find  a primitive $N_2$-th root of unity $\omega_{N_2}$ in $\GF{p^{L_2}}$ in time $\tO(L_2^9 p) = n^{o(1)}$ \cite[Lemma 3.3]{HarveyHL17}.
Note that for any factor $m$ of $N_2$, $\omega_{m}:=\omega_{N_2}^{N_2/m}\in \GF{p^{L_2}}$ is a primitive $m$-th root of unity, and recall the DFT of an length-$m$ array $(a_0,\dots,a_{m-1}) \in (\GF{p^{L_2}})^{m}$ is the array $(\hat a_0,\dots,\hat a_{m-1})\in (\GF{p^{L_2}})^{m}$ where $\hat a_k := \sum_{j=0}^{m-1}a_j\cdot \omega_{m}^{jk}$.

Let $N_2 = m'_1\cdots m'_{d'}$ as in \cref{cor:combined}, where $m'_i \in L_2^{\Theta(1)} \subseteq (\log \log n )^{\Theta(\log^{(4)} n)}$.
Since we assumed $p\le L_1^{O(1)}$, we have $m_i' L_2\log p \le L_2^{O(1)}L_2\log L_1 \le (\log \log n)^{O(\log^{(4)}(n))} < 0.1\log n$.  We also know $N_2 = \Theta(L_1/L_2) \gg \log n$.
Hence, by greedily grouping the factors $m'_i$, we can get a factorization 
\begin{equation}
    \label{eqn:factorn2}
N_2 = m_1m_2 \cdots m_d \, \text{ where } m_i < \tfrac{0.1\log n}{L_2\log p} < m_im_j \text{ for all } i,j\in [d] \,(i\neq j).
\end{equation}
For each $i\in [d]$, define $t_i = \lfloor \frac{0.1\log n}{m_iL_2 \log p} \rfloor \ge 1$.
In the following, we will pack $t_i$ instances of degree-$m_i$ DFTs over $\GF{p^{L_2}}$ in a machine word.

\begin{lemma}
\label{lem:parallel-DFT}
For each $i\in [d]$, after $n^{0.2+o(1)}$ time pre-processing, 
computing $t_i$ instances of degree-$m_i$ DFTs over $\GF{p^{L_2}}$ can be done in $O(1)$ time (assuming a compactly represented input and output).

Similarly, this also holds for the task of multiplying degree-$m_i$ polynomials.
\end{lemma}
\begin{proof}
   The number of such $t_i$ instances of degree-$m_i$ polynomials over $\GF{p^{L_2}}$ is $((p^{L_2})^{m_i})^{t_i} \le  n^{0.1}$ by the definition of $t_i$, so we can preprocess a look-up table in $n^{0.1+o(1)}$ time, which later allows us to look up the DFT answers in $O(1)$ time in word RAM with $\Theta(\log n)$-bit words (assuming the inputs and outputs are packed into $O(1)$ words). 
   A similar argument applies to the task of computing $t_i$ instances of degree-$m_i$ polynomial multiplication, which takes preprocessing time $n^{0.2+o(1)}$.
\end{proof}

Now we describe our second-level algorithm.

\begin{lemma}[Second level]
\label{lem:second-level}
After $n^{0.2+o(1)}$ time pre-processing,  multiplying two degree-$L_1$ polynomials over $\GF{p}$ can be done in $O((L_1\log L_1)\cdot \frac{\log p}{\log n})$ time.   
\end{lemma}
\begin{proof}
Assume $p\le L_1^{O(1)}$; otherwise use \cref{lem:second-level-degenerate} instead.
Recall $ L_1' \in [L_1, 2L_1]$ and $L_2=L_1'/N_2 = \Theta(L_1/N_2)$.
Hence, we first reduce the task of multiplying two degree-$L_1$ polynomials over $\GF{p}$ to $O(1)$ instances of multiplications of two degree-$\lfloor (N_2-1)/2\rfloor$ polynomials over $\GF{p^{L_2}}$.
In more details, this is achieved by packing  contiguous $\lfloor L_2 / 2\rfloor$ coefficients from $\GF{p}$ into an element in $\GF{p^{L_2}}$ (where we divided by two so that the products will not overflow modulo the irreducible monic polynomial of degree $L_2$). This way, the problem becomes the multiplication of two  polynomials  over $\GF{p^{L_2}}$ of degree $ L_1/\lfloor L_2/2\rfloor  = O(N_2)$, which can be easily reduced to  $O(1)$ instances of multiplications of two degree-$\lfloor (N_2-1)/2\rfloor$ polynomials over $\GF{p^{L_2}}$. 

In the following, we describe how to perform this multiplication, whose product should be a  polynomial over $\GF{p^{L_2}}$ of degree at most $N_2-1$.

Recall $N_2$ has a smooth factorization $N_2 =\prod_{i=1}^d m_{i}$ given by \cref{eqn:factorn2}, and recall that we computed a primitive $N_2$-th root of unity $\omega_{N_2} \in \GF{p^{L_2}}$. 
Hence we can use the standard Cooley-Tukey FFT algorithm of length $N_2$ to do the multiplication (see e.g., \cite[Section 2.3]{HarveyHL17}). In the following, we first recall the DFT procedure, and later describe the implementation details in word RAM. 
\vspace{-0.3cm}
\paragraph*{The DFT algorithm.}
Given input array $(a_0,\dots,a_{N_2-1})\in (\GF{p^{L_2}})^{N_2}$, we initialize the working array $A:=(a_{rev(0)},\dots,a_{rev(N_2-1)}) $  where $rev(\cdot)$ is a permutation defined as follows (analogous to the bit-reversal permutation used in the radix-2 version): if $x = \sum_{i=1}^d x_i\cdot  m_1m_2\cdots m_{i-1}$ (where $0\le x_i<m_i$), then $rev(x):= \sum_{i=1}^d x_i \cdot m_{i+1}\dots m_{d-1}m_d$.
Then we perform $d$ rounds of computation on the working array $A$, where in the $i$-th round $(1\le i \le d)$ we perform the following operations (denote $M_i= m_1m_2\cdots m_i$):
\begin{enumerate}
    \item\label{item:second-level-multiplication} 
    For each $0\le k < N_2/M_i$ and $0\le j < M_{i-1}$, let $l =kM_i + j$, and for all $0\le s < m_i$, multiply the ``twiddle factors'':
    \[ A[l+sM_{i-1}] \gets A[l+sM_{i-1}]\cdot \omega_{M_i}^{sj}.\]
    
   In total there are   $N_2$ scalar multiplications over $\GF{p^{L_2}}$ in this round.
    
    \item\label{item:second-level-DFT}  For each $0\le k < N_2/M_i$ and $0\le j < M_{i-1}$, let $l =kM_i + j$, and perform a length-$m_i$ in-place DFT: \[(A[l], A[l  + M_{i-1}], ,\dots, A[l  + (m_i-1) M_{i-1}]) \gets DFT(A[l], A[l  + M_{i-1}], ,\dots, A[l  + (m_i-1) M_{i-1}]).\]
    In total there are $N_2/m_i$ instances of length-$m_i$ DFT over $\GF{p^{L_2}}$ in this round.
\end{enumerate}
One can verify that after $d$ rounds, the working array $A$ becomes the correct DFT result, i.e.,  $A[k] = \sum_{j=0}^{N_2-1} a_j\cdot \omega_{N_2}^{jk}$.

\paragraph*{Implementation of DFT.}
To implement the DFT algorithm described above, we always use a compact representation of the working array $A\in (\GF{p^{L_2}})^{N_2}$ into $O( \frac{N_2 L_2\log p}{\log n})$ words, and we need to use bit parallelism to speed up these operations.
\begin{itemize}
    \item \cref{item:second-level-multiplication} (multiplying the ``twiddle factors''):

In constant time, we multiply the twiddle factors to 
$\Theta(\frac{\log n}{L_2\log p})$ contiguous elements (represented in $O(1)$ words) in the working array $A$
using table look-up (similar to \cref{lem:parallel-DFT} with $m_i$ set to $1$).  (Note that $\frac{\log n}{L_2\log p} \le N_2$.)  
In order to do this table look-up, we also need to prepare a compact representation of the $\Theta(\frac{\log n}{L_2\log p})$ twiddle factors applied to the working array. Note that these twiddle factors are fixed in the algorithm and do not depend on the input, so we can 
 pre-compute the compact representations of them in $\poly(N_2\cdot L_2\log p) \le n^{o(1)}$ time. 

  The total time for \cref{item:second-level-multiplication} over all $d$ rounds (ignoring preprocessing) is $O(d\cdot (N_2\cdot L_2 \log p)/\log n)$.

\item  \cref{item:second-level-DFT} (length-$m_i$ DFTs):

In the $i$-th round, we need to apply length-$m_i$ DFTs on the working array, and we want to speed them up by using \cref{lem:parallel-DFT} to perform $t_i$ DFTs in a batch in constant time. To do this, we need to first collect the array elements $A[l], A[l+M_{i-1}],\dots,A[l+(m_i-1)M_{i-1}]$ participating in each DFT  into a contiguous range of memory in compact representation. (Note that we only need to do this when $i\ge 2$; for $i=1$, since $M_{i-1}=1$, these elements are already in a contiguous range.) More specifically, we need to permute array $A$ into $A'$ so that 
\begin{equation}
    \label{eqn:permu}
 A'[kM_i + j m_i + s] = A[kM_i + j + sM_{i-1}]   \text{ for all } 0\le k<N_2/M_i, 0\le j<M_{i-1}, 0\le s<m_i.
\end{equation}
In other words, if we view the length-$N_2$ working array $A$ as $N_2/M_i$ chunks each representing an $m_i \times M_{i-1}$ matrix in row-major order, then $A'$ is obtained by transposing these matrices into column-major order. 
After permuting $A$ into $A'$, we can perform the required DFTs on $A'$ with time complexity linear in the number of words using the look-up tables from \cref{lem:parallel-DFT}, and then we permutate them back by running the transposition step in reverse. Note that the running time for performing DFTs on $A'$ is dominated by the transposition steps.

 Transposing an $m_i \times M_{i-1}$ matrix can be done by a divide-and-conquer algorithm (similar to \cite[Lemma 9]{journals/jal/Thorup02}) 
  with recursion depth $\log(m_i)$: we start with $m_i$ length-$M_{i-1}$ lists each corresponding to a leaf of the recursion tree, and at each internal node of the recursion tree we interleave the lists returned by its two child nodes.
  Here, using word operations (which can be replaced by table look-ups after preprocessing), interleaving two lists can be done with time complexity linear in the number of words in their compact representations. Hence, transposing an $m_i \times M_{i-1}$ matrix (with entries from $\GF{p^{L_2}}$) via divide-and-conquer takes total time $\sum_{q=0}^{\log m_i} 2^q\cdot (O(\frac{(m_i/2^q) M_{i-1} L_2\log p}{\log n}) + O(1) ) \le O((\log m_i)m_i M_{i-1}L_2\frac{\log p}{\log n} + m_i)$, and transposing $N_2/M_i$ such matrices in total takes $O((\log m_i)N_2L_2 \frac{\log p}{\log n} + N_2/M_{i-1})$ time. For $i\ge 3$, we have $M_{i-1}\ge m_1m_2 > \frac{0.1\log n}{L_2\log p}$ from \cref{eqn:factorn2}, and the second term $N_2/M_{i-1}$ in the time complexity is dominated, so the run time becomes $O((\log m_i)N_2L_2 \frac{\log p}{\log n})$. For the  remaining case $i=2$, the same run time can be achieved by  slightly modifying the divide-and-conquer matrix transposition algorithm: when the total bit length of the lists in the current recursion subtree is below $0.1\log n$, we simply look up the transposition result from a preprocessed table instead of recursing.

  To summarize, the total time for
  \cref{item:second-level-DFT} for 
   all $d$ rounds is   
\begin{align*}
& O\left(\sum_{i=1}^d\log(m_i) \cdot N_2 L_2\frac{\log p}{\log n}\right)\\
& = O(\log N_2) \cdot N_2 L_2\frac{\log p}{\log n} \tag{by \cref{eqn:factorn2}}\\
& = O(\log N_2) \cdot L_1\frac{\log p}{\log n} \tag{by $L_2=L_1'/N_2$ and $L_1'=\Theta(L_1)$}\\
& = O\left (L_1\log L_1 \cdot \frac{\log p}{\log n}\right ).
\end{align*}
\end{itemize}

 Finally, note that the initialization step (applying the $rev(\cdot)$ permutation to the input array) can be done in a similar fashion to the transposition steps described in \cref{item:second-level-DFT}, with the same total time complexity  $O\left (L_1\log L_1 \cdot \frac{\log p}{\log n}\right )$.

Note that the total time of \cref{item:second-level-multiplication} is dominated by \cref{item:second-level-DFT}, so the total time complexity of the algorithm is
$O\left (L_1\log L_1 \cdot \frac{\log p}{\log n}\right )$.
The total pre-processing time of calling \cref{lem:parallel-DFT} $d$ times is $O(n^{0.2+o(1)} \cdot d) = n^{0.2+o(1)}$, and the pre-processing time for other look-up tables used by the algorithm can also be bounded similarly by $n^{0.2+o(1)}$. 
\end{proof}

The proof of \cref{lem:second-level} described above can also prove the following slightly stronger statement:
\begin{cor}
    \label{cor:varl1}
    Let $\tilde L_1 \in [L_1^{0.2}, L_1^{10}]$ be a power of two.
After $n^{0.2+o(1)}$ time pre-processing,  multiplying two degree-$\tilde L_1$ polynomials over $\GF{p}$ can be done in $O((\tilde L_1\log \tilde L_1)\cdot \frac{\log p}{\log n})$ time.   
\end{cor}
\begin{proof}
The only bounds on $L_1$ that we used in proving \cref{lem:second-level} are $L_1 \in (\log n)^{\Theta(\log \log \log n)}$ and $p\le L_1^{O(1)}$, which also hold for $\tilde L_1$.
Hence we can simply repeat the proof of \cref{lem:second-level} with $\tilde L_1$ in place of $L_1$. 
\end{proof}

Finally we describe the first level of our algorithm, proving \cref{lem:packed_FFT}.
\begin{proof}[Proof of \cref{lem:packed_FFT}]
Recall $p\le n^{0.01}$, $L_1=n'/N_1$ and $n=\Theta(n')$.
By the same reasoning as in the proof of \cref{lem:second-level},  we can reduce the task of multiplying two degree-$n$ polynomials over $\GF{p}$ to $O(1)$ instances of polynomial multiplication over $\GF{p^{L_1}}$ whose product has degree at most $N_1-1$.
Note that we can find a representation of $\GF{p^{L_1}}$ by finding an irreducible monic
polynomial of degree $L_1$, which can be done in  $\tO(L_1^4 p^{1/2}) = n^{0.005+o(1)}$ time deterministically \cite{Shoup88}. 
Since we have shown earlier that $N_1\mid  p^{L_1}-1$, we can find  a primitive $N_1$-th root of unity $\omega_{N_1}\in \GF{p^{L_1}}$ in $\tO(L_1^9 p) = n^{o(1)}\cdot n^{0.01}$ time \cite[Lemma 3.3]{HarveyHL17}.
Let $N_1 = \prod_{i=1}^d m_{i}$ as in \cref{cor:combined}, where $\sqrt{L_1}/2 \le m_{i}\le  L_1^{3}$.

To do this multiplication, we run Cooley-Tukey FFT using this smooth $N_1$-th root. Similar as the proof of \cref{lem:second-level}, it involves $d$ rounds of computation on a (compactly represented) working array of $N_1$ elements from $\GF{p^{L_1}}$, where the $i$-th round involves the following operations.
\begin{enumerate}
    \item $N_1$ scalar multiplications over $\GF{p^{L_1}}$:  multiply the ``twiddle factors'' to each of the $N_1$ elements in the array. 
    
    The cost for preparing all possible twiddle factors $\{\omega_{N_1}^j\}_{j\in [N_1]}$ is $N_1$ scalar multiplications over $\GF{p^{L_1}}$, which can be absorbed into the cost of this step. (In contrast to the proof of \cref{lem:second-level}, here we do not need to prepare the compact representations of multiple twiddle factors packed into one word, since here each twiddle factor already occupies more than one word.)  
    \item $N_1/m_i$ instances of length-$m_i$ DFT over $\GF{p^{L_1}}$.
\end{enumerate}

Let $T_{D, L_1}(\ell)$ denote the cost of computing the DFT of a length-$\ell$ polynomial over $\GF{p^{L_1}}$, and let $T_{M, L_1}$ denote the cost of scalar multiplication over $\GF{p^{L_1}}$. Then the total time complexity for FFT is (up to constant factors)
\[ 
    \sum_{i=1}^d  \big (N_1 \cdot T_{M, L_1} + \frac{N_1}{m_i} T_{D, L_1}(m_i)  \big ).
\]
Now we analyze the two terms separately.

\begin{itemize}
    \item 
To analyze $T_{M,L_1}$, note that a scalar multiplication over $\GF{p^{L_1}}$ can be done by computing the product of two degree-$L_1$ polynomials over $\GF{p}$, and then mapping it back to $\GF{p^{L_1}}$ by reducing modulo a degree-$L_1$ monic irreducible polynomial over $\GF{p}$.
By \cref{lem:second-level}, multiplying two degree-$L_1$ polynomials over $\F_p$ can be done in $O((L_1\log L_1)\frac{\log p}{\log n})$ time.
Using Newton's iteration (see e.g., \cite[Section 9]{von2013modern}),
degree-$L_1$ polynomial division with remainder can be reduced to 
$O(\log L_1)$ instances of polynomial multiplication with degrees $L_1,\frac{L_1}{2},\frac{L_1}{4},\frac{L_1}{8},\dots$ respectively.
For multiplication with degree $\ge L_1^{0.2}$, we invoke \cref{cor:varl1}. For smaller degree, we use brute-force quadratic-time multiplication. The total time for degree-$L_1$ polynomial division is thus (up to a constant factor) 
\begin{align*}
  \sum_{j=0.2\log_2 L_1}^{\log_2 L_1} (2^j \log 2^j)\tfrac{\log p}{\log n} + \sum_{j=0}^{0.2\log_2 L_1} (2^j)^2  \le  O(L_1 \log L_1)\cdot \tfrac{\log p}{\log n} + O(L_1^{0.4})
   =  O(L_1 \log L_1)\cdot \tfrac{\log p}{\log n}.
\end{align*}
Hence, $T_{M,L_1} = O((L_1\log L_1)\frac{\log p}{\log n})$.

\item
To analyze 
$T_{D, L_1}(m_i)$, we use Bluestein's chirp transform (see \cite[Section 2.5]{HarveyHL17}) to reduce the task of computing a length-$m_i$ DFT over $\GF{p^{L_1}}$ to multiplying two degree-$m_i$ polynomials over $\GF{p^{L_1}}$. This can further be reduced to multiplying  degree-$2 m_i L_1$  polynomials over $\GF{p}$ via Kronecker substitution (see \cite[Section 2.6]{HarveyHL17}), which can be solved using \cref{cor:varl1} (recall $m_i\le L_1^{3}$) in time $O(m_i L_1 \cdot \log (m_i L_1)\cdot  \frac{ \log p}{\log n})$. 
Afterwards, we divide $m_i$ degree-$2L_1$ polynomials over $\GF{p}$ by the degree-$L_1$ irreducible monic polynomial over $\GF{p}$, to map the elements back to $\GF{p^{L_1}}$, in  total time $O(m_i L_1 \cdot \log(L_1) \cdot \log p / \log n)$ (similar to the previous  paragraph). Hence, $T_{D,L_1}(m_i) \le O(m_i L_1  \log (L_1)\cdot  \frac{ \log p}{\log n})$ 
\end{itemize}

Hence, the total time becomes (up to constant factors)
\begin{align*}
 &    \sum_{i=1}^d  \big (N_1 \cdot T_{M, L_1} + \frac{N_1}{m_i} T_{D, L_1}(m_i)  \big )\\
    \le \ &  \sum_{i=1}^d  \big ( N_1 (L_1 \log L_1) \frac{\log p}{\log n} + \frac{N_1}{m_i} m_iL_1 \log(L_1) \cdot \frac{\log p}{\log n} \big ) \\
    \le \ &   O\big (d N_1 (L_1 \log L_1)\frac{\log p}{\log n} \big) \\
    \le \ &   O\big (d \cdot n \log L_1 \cdot \frac{\log p}{\log n}\big ).
\end{align*}
Recall that \cref{cor:combined} gave the factorization $N_1 = \prod_{i=1}^d m_{i}$ with $m_{i}\in  L_1^{\Theta(1)}$, so $d\log L_1 = \Theta( \log N_1) = \Theta(\log n)$, and the final run time becomes $O(n\log p)$ as desired.
\end{proof}

\end{document}